\documentclass{lmcs_modified} %

\usepackage{multicol}
\usepackage{algorithm}
\usepackage{bm}
\usepackage{algpseudocode}
\usepackage{tikz,tikz-qtree}
\usepackage{amssymb}
\usetikzlibrary{positioning, graphs}
\usepackage[bibliography=common]{apxproof}

\usepackage{hyperref}
\usepackage{cleveref}

\newcommand{\pal}{\mathrm{PAL}}
\newcommand{\pale}{\pal_{>0}}

\newcounter{mycount}

\newcommand{\dom}{\mathrm{dom}}

\makeatletter
\newcommand{\problemtitle}[1]{\gdef\@problemtitle{#1}}%
\newcommand{\probleminput}[1]{\gdef\@probleminput{#1}}%
\newcommand{\problemquestion}[1]{\gdef\@problemquestion{#1}}%
\NewEnviron{problem}{
  \problemtitle{}\probleminput{}\problemquestion{}%
  \BODY%
  \par\addvspace{.5\baselineskip}
  \noindent
  \begin{tabularx}{\textwidth}{@{\hspace{\parindent}} l X c}
    \multicolumn{2}{@{\hspace{\parindent}}l}{\@problemtitle} \\% Title
    \textbf{Input:} & \@probleminput \\% Input
    \textbf{Question:} & \@problemquestion%
  \end{tabularx}
  \par\addvspace{.5\baselineskip}
}

\newcommand{\NN}{\mathbb{N}}
\newcommand{\ZZ}{\mathbb{Z}}

\newcommand{\LL}{\text{L}}
\newcommand{\RR}{\text{R}}
\renewcommand{\prob}{\text{\#pM}}
\newcommand{\probone}[1]{\prob(#1)}

\renewcommand{\SS}{\mathrm{S}}

\newcommand{\prim}{\text{prim}}

\newcommand{\compl}{\complement}

\newcommand{\calA}{\mathcal{A}}
\newtheoremrep{theorem}{Theorem}[section]       %
\newtheoremrep{proposition}[theorem]{Proposition}
\newtheoremrep{lemma}[theorem]{Lemma}
\newtheoremrep{example}[theorem]{Example}
\newtheoremrep{definition}[theorem]{Definition}
\newtheoremrep{corollary}[theorem]{Corollary}
\newtheoremrep{conjecture}[theorem]{Conjecture}
\newtheoremrep{claim}[theorem]{Claim}
\newtheoremrep{remark}[theorem]{Remark}
\newtheoremrep{fact}[theorem]{Fact}

\AddToHook{env/proposition/begin}{\crefalias{theorem}{proposition}}
\AddToHook{env/lemma/begin}{\crefalias{theorem}{lemma}}
\AddToHook{env/example/begin}{\crefalias{theorem}{example}}
\AddToHook{env/definition/begin}{\crefalias{theorem}{definition}}
\AddToHook{env/corollary/begin}{\crefalias{theorem}{corollary}}
\AddToHook{env/conjecture/begin}{\crefalias{theorem}{conjecture}}
\AddToHook{env/claim/begin}{\crefalias{theorem}{claim}}
\AddToHook{env/remark/begin}{\crefalias{theorem}{remark}}
\AddToHook{env/fact/begin}{\crefalias{theorem}{fact}}

\newcommand{\nruns}{\#\text{runs}}
\newcommand{\restr}[2]{#1_{|#2}}

\theoremstyle{plain} %

\begin{document}

\title{On the Complexity of Language Membership for Probabilistic Words}
\author[Antoine Amarilli]{Antoine Amarilli\lmcsorcid{0000-0002-7977-4441}}[a]
\author[Mikaël Monet]{Mikaël Monet\lmcsorcid{0000-0002-6158-4607}}[a]
\author[Paul Raphaël]{Paul Raphaël\lmcsorcid{0000-0002-1825-0097}}[b]
\author[Sylvain Salvati]{Sylvain Salvati\lmcsorcid{0000-0002-1825-0097}}[a]

\address{Univ. Lille, Inria, CNRS, Centrale Lille, UMR 9189 CRIStAL, France}	%
\email{a3nm@a3nm.net, mikael.monet@inria.fr, sylvain.salvati@inria.fr}  %

\address{École normale supérieure, Paris, France}	%
\email{paul.raphael@ens.psl.eu}  %

\begin{abstract}
  \noindent We study the membership problem to context-free languages~$L$ (CFLs) on
  \emph{probabilistic words}, that specify for each position a probability
  distribution on the letters (assuming independence across positions). Our task
  is to compute, given a probabilistic word, what is the probability that a word
  drawn according to the distribution belongs to~$L$. This
  problem generalizes the problem of counting how many words of length $n$
  belong to~$L$, or of counting how many
  completions of a partial word belong to~$L$.

  We show that this problem is in
  polynomial time for unambiguous context-free languages (uCFLs), but can be \#P-hard
  already for unions of two linear uCFLs. More generally, we show that the
  problem is in polynomial time for so-called \emph{poly-slicewise-unambiguous languages},
  where given a length~$n$ we can tractably compute an uCFL for the words of
  length~$n$ in the language. This class includes some inherently ambiguous
  languages, and implies the tractability of \emph{bounded CFLs} and of
  languages recognized by \emph{unambiguous polynomial-time counter automata}; but we show
  that the problem can be \#P-hard for nondeterministic counter automata, even
  for Parikh automata with a single counter. We then introduce classes of
  circuits from knowledge compilation which we use for tractable counting, and
  show that this covers the tractability of poly-slicewise-unambiguous languages
  and of some CFLs that are not poly-slicewise-unambiguous. Extending these
  circuits with negation further allows us to show tractability for the language
  of primitive words, and for the language of concatenations of two palindromes. 
  We also show that, when the target language is given as input, the
  probabilistic membership problem is intractable already when the 
  language asks whether there is a factor that matches one partial word;
  however, it
  becomes tractable when the language is given as a $k$-ambiguous
  automaton for any fixed $k>0$, contrasting with the case of context-free
  grammars (CFGs).
  We finally show the conditional undecidability of the meta-problem that asks,
  given a CFG, whether the probabilistic membership problem for that CFG is
  tractable or \#P-hard.
\end{abstract}

\maketitle

\section{Introduction}
\label{sec:intro}
One of the most fundamental problems in formal languages is
the \emph{membership problem} of determining whether an input word belongs to
a language of interest. The problem is in linear time for regular
languages when given as deterministic finite automata, but complexity increases
for the class of \emph{context-free languages} (CFLs): the best algorithms on an
input word of length~$n$ run in time $O(n^\omega)$ where $\omega$ is the
exponent of Boolean matrix multiplication~\cite{abboud2018if}. Because of the
practical importance of CFL parsing, several classes of CFLs have been
identified with a better complexity. For instance, \emph{unambiguous CFLs} (uCFLs),
which can be recognized by \emph{unambiguous context-free grammars} (uCFGs), can
be parsed in $O(n^2)$~\cite{earley1970efficient}; further, \emph{deterministic CFLs}
can be parsed in linear time.

Beyond the question of membership testing for a single input word, subsequent
works have investigated more general tasks. For instance, one problem is
\emph{counting}: given a CFL $L$ and a
length $n \in \NN$, determine how many words of length~$n$ are in~$L$.
(Counting problems that do not fix the length can also be defined, e.g., the
\emph{coin-flip measure} of~\cite{clemente2020complexity}.)
An FPRAS to \emph{approximately} count $|L \cap \Sigma^n|$ when given $n$ and a CFG
for $L$ was recently claimed
in~\cite{meel2026cfg}, but in the exact setting the task is intractable already for
regular languages presented as nondeterministic
automata (NFAs)~\cite{alvarez1993very}. This counting
problem was also
studied for CFLs in~\cite{bertoni1991complexity} when the target language is
fixed (and the input only specifies the length); we will review this work in
more detail later.
Another task is to \emph{sample} a word of the language of~$L$ uniformly at random:
a polynomial algorithm for uCFGs was shown
in~\cite{hickey1983uniform,mairson1994generating}, and a generalization to
finitely ambiguous CFGs is given in~\cite{bertoni2001random}, along with a
quasi-polynomial time algorithm for almost uniform
sampling~\cite{gore1997quasi}. Other tasks include \emph{enumeration},
i.e., producing the sequence of all words in a CFL~$L$ according to some order
(e.g., the lexicographic order):
see~\cite{semba1981generation,makinen1997lexicographic,domosi2000unusual,costa2015naive};
and \emph{ranking}, i.e., determining how many words in~$L$ are smaller than an
input word~\cite{huynh1988complexity,huynh1990complexity,bertoni1991ranking}.

In this paper, we consider the problem of \emph{language membership for probabilistic words},
for short \emph{probabilistic membership}.
In this 
setting, a \emph{probabilistic word} on an alphabet $\Sigma$ is a sequence $p =
p_1 \cdots p_n$ of length $n$, where each $p_i$ is a probability distribution
on~$\Sigma$. The probabilistic word $p$ is a concise representation of a
probability distribution of words (each of length~$n$) obtained by picking the
$i$-th letter according to the distribution~$p_i$, and assuming independence
across letters. 
Probabilistic words are also called 
\emph{weighted sequences} or \emph{weighted strings} in the literature, with applications
in
bioinformatics~\cite{atteson1998calculating,kleffe1990exact,nicodeme1999motif,lladser2008multiple}
and pattern recognition as well as data mining~\cite{charalampopoulos2020property,radoszewski2020streaming,christodoulakis2004pattern}.
For a fixed language~$L$, the \emph{probabilistic membership}
problem asks us to determine, given a probabilistic word~$p$, the
total probability of the outcomes that belong to~$L$, i.e., the probability that
a word drawn according to~$p$ belongs to~$L$. The problem can be solved naively
by going over all possible outcomes (of which there can be up to $|\Sigma|^n$);
our goal is to understand for which languages~$L$ we can solve this problem
more efficiently, e.g., in polynomial time in~$n$.

The probabilistic membership problem of course generalizes
non-probabilistic membership ,
and it can be seen as a weighted generalization of the
counting problem of~\cite{bertoni1991complexity}; the latter corresponds to the
case of probabilistic words that can only use the uniform distribution
over~$\Sigma$ at every position. 
Moreover, it is
not hard to see that the tractability of probabilistic membership
can be used as a subroutine to allow efficient sampling, enumeration, and
ranking in the radix order (up to polynomial-time overhead). Indeed,
probabilistic membership generalizes the question of counting how many words
of length $n$ start with a given prefix, so that sampling, enumeration, and
ranking can be performed using self-reducibility.
Furthermore, we can also see probabilistic words as a weighted generalization of
so-called \emph{partial words}: these are words on an
alphabet $\Sigma$ augmented with a wildcard~'?' that can be replaced by
any letter of~$\Sigma$. 
Partial words have been studied in their own right
\cite{berstel1999partial,blanchet2007algorithmic,manea2010hard,manea2013hardness,blanchetsadri2009counting},
and
we can see them as the special case of
probabilistic words where each distribution is either uniform or Dirac (i.e.,
with only one possible outcome).
Probabilistic membership to a language~$L$ is thus a generalization of the problem of counting the
completions of partial words that belong to~$L$.

This paper focuses on the probabilistic membership problem,
and
aims at understanding for which CFLs it can be solved in polynomial time.
Borrowing terminology from database
theory~\cite{vardi1982complexity}, we mostly study the problem in \emph{data
complexity}, i.e., we assume that the target language is fixed and measure
complexity only as a function of the input word. Some of our
results, however, are about the \emph{combined complexity} of the
problem, i.e., when the input contains both the word and some
representation of the target language (e.g., as a CFG or NFA). 
Our work does
not achieve a complete characterization of the CFLs that enjoy tractable
probabilistic membership, but we give a hierarchy of sufficient conditions for
tractability of the task (even beyond CFLs), complemented by lower bounds. We
summarize our contributions below.

\paragraph{Contributions and outline}
We give preliminaries and formally introduce the
problem in \cref{sec:problem}. We then give some first
results in \cref{sec:ambiguity}, which are essentially generalizations of the
unweighted counting results of~\cite{bertoni1991complexity}.
Namely, we show that probabilistic membership is tractable for uCFLs (even in
combined complexity), because we
can count the total probability of derivation trees with a variant of
the Cocke-Younger-Kasami (CYK) algorithm.
However, we show that the task can be
intractable (namely, \#P-hard) for some inherently ambiguous CFLs: we give a
simple example of such a CFL, and generalize and extend techniques
from~\cite{bertoni1991complexity} to show hardness even in the case of the union
of two unambiguous \emph{linear CFLs}, i.e., the right-hand-side of each production rule
contains at most one nonterminal.

The next section (\cref{sec:combined}) shows that there are CFLs that are
inherently ambiguous but for which probabilistic membership is nevertheless
tractable. These include, for instance, all \emph{bounded}
CFLs~\cite{ginsburg1964bounded}: these are
known to coincide with the class of \emph{polyslender
CFLs}~\cite{ilie2000characterization}, i.e., those CFLs
containing only polynomially many words of length~$n$. Thus, we can tractably
solve probabilistic membership for bounded CFLs~$L$ (even in combined
complexity)
by simply listing all words
of length~$n$ of~$L$ explicitly and summing their probabilities. Bounded CFLs are incomparable to uCFLs,
as there are non-bounded uCFLs (e.g., palindromes) and also inherently
ambiguous bounded CFLs (e.g., $\{a^n b^m c^m \mid n, m \in \NN\} \cup \{a^n b^n
c^m \mid n, m \in \NN\}$). We give a unifying explanation to the
tractability of uCFLs and bounded CFLs by introducing the class of languages $L$
that are \emph{poly-slicewise-unambiguous} in the sense that, given a length $n
\in \NN$, we can compute in polynomial time in~$n$ a uCFG $G$ recognizing the
words of length~$n$ of~$L$, i.e., recognizing a language $L'$ with $L \cap
\Sigma^n = L' \cap \Sigma^n$. For such languages,
probabilistic membership is tractable simply by reducing to uCFGs. We last introduce a
formalism of \emph{unambiguous polynomial-time counter automata} which are
poly-slicewise-unambiguous: these recognize in particular some CFLs that are
neither bounded nor unambiguous.
We show that probabilistic membership is tractable for such automata, but that
it can be \#P-hard for some
CFLs accepted by very restricted nondeterministic counter automata.

The natural question is then whether poly-slicewise-unambiguity is the unifying
explanation for the tractability of probabilistic membership for CFLs. We answer
in the negative in \cref{sec:circuits} by introducing a framework of
\emph{tractable circuits} from the field of \emph{knowledge
compilation}~\cite{darwiche2002knowledge}, inspired by database
theory~\cite{amarilli2024tractable}. Intuitively, tractable circuits are
circuits that give a factorized representation of a set of solutions 
(intuitively corresponding to outcomes of a probabilistic word) using some
tractable operators: \emph{Cartesian product} (called \emph{decomposable AND} in the
setting of Boolean knowledge compilation circuit classes), and \emph{disjoint
union} (called \emph{deterministic OR} in the Boolean setting). These circuits
are a multivalued analogue to so-called \emph{d-DNNF
circuits}~\cite{darwiche2002knowledge}. The choice of
operators ensures that we can compute in linear time (up to the cost of
arithmetic operations) the total probability of the assignments captured by a
tractable circuit. We show that these circuits can be used to recapture the
tractability of all poly-slicewise-unambiguous languages. What is more, we show
that there are CFLs that are not poly-slicewise-unambiguous but admit tractable
circuits: this is the case of the language recently studied in
\cite{kimelfeld2025formal} whose length-$n$ slices were 
shown~\cite{mengel2025lower} not to admit uCFLs of polynomial size in~$n$.

We then continue our study of tractable circuits in \cref{sec:palindromes} and
extend them by adding a complementation operator (corresponding to negation in
the Boolean sense). This follows the observation from probabilistic
databases~\cite{monet2020solving} that counting the satisfying valuations of
\mbox{d-DNNF} circuits %
can be tractably achieved even in the presence of negation. Further, negation is
very useful to encode some inclusion-exclusion-based counting arguments (through
the full extent to which this is true is not yet understood
\cite{amarilli2024non}). Tractable circuits with negation can be used in
particular to generalize the obvious fact
that probabilistic membership is tractable for the complements of all tractable
languages studied so far (which are generally not CFLs). We use those circuits
to further show tractability for some other languages: we introduce the
technique with the language of \emph{primitive words} (whose context-freeness is
a long-standing open problem~\cite{ito2014context}), and further apply it to
the language of the \emph{concatenation of palindromes}, a
non-linear CFL of unbounded degree of ambiguity~\cite{crestin1972langage}. Our
tractability result for probabilistic membership for this language generalizes
an earlier result on unweighted counting~\cite{kemp1982number}.

We then revisit in \cref{sec:morecombined} the question of combined complexity,
i.e., the complexity of probabilistic membership when we are given as input both
a probabilistic word and a representation of the target language. We focus 
on the case where the target language is a regular languages represented as
an NFA. Probabilistic membership
is intractable in combined complexity already in this
setting when we allow arbitrary NFAs as input~\cite{alvarez1993very}, but in this section we show two additional
results when making assumptions on the input automaton. First, we show that
hardness already holds in the case where the target language is in a very
restricted class of \emph{partial pattern matching} problems, i.e., we ask for
the probability of the outcomes that contain a factor matching a
partial pattern given as input.
This result nicely complements another existing hardness result from the bioinformatics literature~\cite{ma2007complexity}; we discuss this in more details in \cref{sec:morecombined}.
Second, we show that,
by contrast, for any fixed constant $k>0$ the probabilistic membership problem can be
tractably solved in combined complexity when the input NFA is assumed to be
\emph{$k$-ambiguous}. This contrasts with the case of context-free grammars,
as we showed in \cref{sec:ambiguity} that probabilistic membership is
intractable for 2-ambiguous CFGs even in data complexity.

We finally conclude the paper in \cref{sec:conclusion} and survey directions for
future work; we also mention there the undecidability of finding out,
given a CFG, whether it admits PTIME probabilistic membership.

This paper is the extended version of the conference
article~\cite{amarilli2026complexity}. Compared to the conference version, the
present paper contains the complete proofs of the results in the article, we added
\cref{sec:morecombined} with additional results on the combined complexity
of the problem on NFAs, and we added pointers to related work from the field of
bioinformatics that we were not previously aware
of (such as~\cite{kleffe1990exact,atteson1998calculating,ma2007complexity}).

\section{Preliminaries and Problem Statement}
\label{sec:problem}
\paragraph{Words and languages}
We fix a non-empty finite set of symbols $\Sigma$ called an \emph{alphabet} and
whose elements are called \emph{letters}. 
We let $\Sigma^*$ be the set of all \emph{words} over~$\Sigma$, which are finite
sequences of letters of~$\Sigma$.
We write $\epsilon$ for the empty word.
The \emph{mirror image} of $w = a_1 \cdots a_n \in \Sigma^*$ is $w^R \coloneq
a_n \cdots a_1$; we call $w$ a \emph{palindrome} if $w = w^R$.
A \emph{language} over $\Sigma$ is a subset of $\Sigma^*$. %
We omit the definition of standard language operations, e.g., concatenation,
Kleene star, etc.

\paragraph{Grammars}
 A \emph{context-free grammar} (CFG) over $\Sigma$
 is a tuple
$\Gamma = (N, S, P)$ where
$N$ is a non-empty finite set of symbols called \emph{nonterminals}, 
$S \in N$ is the \emph{axiom},
and $P \subseteq N \times (\Sigma \cup N)^*$ is a set of \emph{production
rules}, with each production
$(U, \alpha) \in P$ written as $U \rightarrow \alpha$.
We call the letters of~$\Sigma$ \emph{terminals}, and write them in lowercase.
We say that $\Gamma$ is a \emph{linear CFG} if the right-hand-side of every
production contains at most one nonterminal.
The \emph{size}~of $\Gamma$, written~$|\Gamma|$, is simply $|N|$
plus the total size of the production rules.

We omit the standard formal definitions of \emph{derivations} and \emph{parse
trees} for CFGs. A CFG~$\Gamma$ \emph{recognizes} the language $\LL(\Gamma)$ formed of the
words $w \in \Sigma^*$ on which $\Gamma$ admits a parse tree. We say that
$\Gamma$ is \emph{unambiguous} (or a \emph{uCFG}) if on every word it admits at
most one parse tree; otherwise $\Gamma$ is \emph{ambiguous}. We call $\Gamma$ \emph{$k$-ambiguous} for $k\in \NN$ if on
every word it admits at most $k$ parse trees. A language $L$ is a \emph{context-free
language} (CFL) if there is a CFG $\Gamma$ with $\LL(\Gamma) = L$. It is an
\emph{unambiguous CFL} (uCFL) if $\Gamma$ can be taken to be a uCFG (but
note that there may also be ambiguous CFGs recognizing $L$). We define
\emph{unambiguous linear CFLs} and 
\emph{\(k\)-ambiguous linear CFLs}
in the expected way.
A CFL $L$ is called \emph{inherently ambiguous} if it is not a uCFL, meaning
that every CFG recognizing~$L$ is ambiguous.

\paragraph{Probabilistic words}
We now define the notion of \emph{probabilistic words} on an
alphabet~$\Sigma$:

\begin{definition}
  A \emph{probability distribution} $p$ on a finite set $S$ is a function from~$S$ to
  $[0,1]$ such that $\sum_{s \in S} p(s) = 1$.
  We call $p$ \emph{uniform} if we have 
  $p(s) \coloneq 1/|S|$ for all $s \in S$, and call~$p$ the
  \emph{Dirac distribution on $s \in S$} if we have
  $p(s) \coloneq 1$ (hence $p(s') = 0$ for all $s' \neq s$).
  A \emph{probabilistic word} is a tuple $p =(p_1,\ldots,p_n)$ of probability
  distributions on $\Sigma$, with~$n$ being the \emph{length} of~$p$. The semantics of
  $p$ is that it represents a probability distribution $p \colon \Sigma^* \to
  [0,1]$, 
  also written $p$ by abuse of notation. Intuitively, $p$ gives
  probability~$0$
  to words of length different from~$n$, and otherwise the probability is obtained by
  multiplying the probabilities given by the distributions to the letters that
  occur at each position.
  Formally, for $w = w_1 \cdots w_m$ in $\Sigma^*$, we have:
  $p(w) \coloneq 0$ if $m \neq n$, and
  $p(w) \coloneq \prod_{k=1}^n p_k(w_k)$ if $m = n$.
\end{definition}

Concatenating probabilistic words means concatenating the
tuples: this is associative and the neutral element is
the empty probabilistic word (the empty tuple).

\paragraph{Probabilistic membership problem}
We now define the problem that we study, called \emph{probabilistic
membership}.
Given a language $L$ on an alphabet $\Sigma$, and a probabilistic word~$p$, we
abuse notation and write $p(L)$ to mean the 
probability that a random word of~$\Sigma^*$ drawn according to~$p$ will belong
to the language~$L$. Formally, we have: $p(L) \coloneq
\sum_{w \in L} p(w)$.
The \emph{probabilistic membership problem} then asks:
given a language $L$ and a probabilistic word~$p$, compute
$p(L)$. Here, probabilities of the input word are given as rational numbers $\frac{r}{q}$ for $(r,q)\in \mathbb{N} \times (\mathbb{N} \setminus \{0\})$, with $r$ and $q$ encoded in binary.
We focus on the \emph{data complexity} perspective,
where the target $L$ is fixed: this defines the \emph{probabilistic membership problem for $L$}, written $\probone{L}$,
with complexity measured as a function of~$p$.
However, some of our tractability results 
hold in \emph{combined complexity},
where 
the input is 
$p$ along with a CFG $\Gamma$ with $\LL(\Gamma) = L$.
We also revisit the question of combined complexity in its own right in
\cref{sec:morecombined}.
Let us now illustrate some special cases of the probabilistic membership problem:

\begin{itemize}
  \item For each $u \in \Sigma^*$, the \emph{Dirac probabilistic word} $u' =
    (p_1, \ldots, p_{|u|})$ is the one that ensures that $u'(u) = 1$, namely,
    $p_i$ is
    the Dirac distribution on~$u_i$ for each $1 \leq i \leq |u|$. We may abuse
    notation and identify $u'$ with~$u$.
    Probabilistic membership to a language~$L$ thus
    generalizes the standard (non-probabilistic) membership problem to~$L$.
  \item For each $n \in \NN$ the \emph{uniform probabilistic word} of length~$n$
    is $p_n = (p_1, \ldots, p_n)$ where each~$p_i$ is the uniform
    distribution on~$\Sigma$. Its outcomes are precisely the words of
    $\Sigma^n$, each with probability $1/|\Sigma|^n$. Probabilistic membership
    to~$L$ thus generalizes (up to renormalization by an easily-computable factor
    of~$|\Sigma|^n$) the unweighted counting problem studied
    in~\cite{bertoni1991complexity}: given an integer~$n \in \NN$, compute the
    cardinality $|L \cap \Sigma^n|$.
  \item Probabilistic words are also a weighted generalization
of \emph{partial words}~\cite{berstel1999partial}. A \emph{partial word} is a
word $p$ over the alphabet $\Sigma \cup \{?\}$, for $?$ a fresh wildcard. The
semantics of $p$ is that it represents a set of possible words over $\Sigma^*$ obtained
by considering all possible \emph{completions} of~$p$, i.e., all possible ways
to replace each occurrence of $?$ in~$p$ by a
letter of~$\Sigma$. Note that all completions have length $|p|$. A partial word
$p$ can be translated to a probabilistic word $q = (q_1, \ldots, q_{|p|})$
where for each $1 \leq i \leq |p|$ we set $q_i$ as the Dirac on~$p_i$ if $p_i
\in \Sigma$ or as the uniform distribution if $p_i = {?}$.
Probabilistic membership to~$L$ thus generalizes the problem of counting the
number of completions of~$p$ that are in~$L$.
In fact, all our hardness results already hold for a variant of
this problem, that we define next.
\end{itemize}

\paragraph{Domain-constrained 
  completion counting problem}

All hardness problems mentioned in the paper will be shown on a slight
generalization of the 
problem of counting the completion of partial words:

\begin{definition}
  \label{def:dcccp}
  Let $\Sigma$ be an alphabet, let $L$ be a language over~$\Sigma$, and let $?$ be a
  fresh symbol not in $\Sigma$. The \emph{completion counting problem} (CCP) for
  the language~$L$ is the following counting problem: given a partial word $u$ 
(i.e., a word over $\Sigma \cup \{?\}$), count the number of
  completions of~$u$ that are in~$L$ (i.e., the number of words obtained
by substituting the occurrences of $?$ in $p$ by letters of~$\Sigma$ to obtain a
word in~$L$).

  Let $\Sigma' \subseteq \Sigma$ be a subalphabet.
  The \emph{domain-constrained 
  completion counting problem} (DCCCP) for~$L$ and~$\Sigma'$ is the following
  counting problem: given a partial word $u$, count the number of
  \emph{$\Sigma'$-completions} of~$u$ that are in~$L$, where a
  \emph{$\Sigma'$-completion} is a completion of~$u$ where each occurrence of
  $?$ must be replaced by a letter of~$\Sigma'$.
\end{definition}

In other words, the completion counting problem for partial words in the
usual sense (i.e., CCP) for a language~$L$ amounts to the DCCCP problem for~$L$
and~$\Sigma' = \Sigma$. All of our hardness results for $\prob$ with languages
$L$ will be shown as hardness results for DCCCP for the same language $L$ and
with some subalphabet $\Sigma'\subseteq \Sigma$ of cardinality~2. (We remark
that, by contrast, our results do not imply the intractability of the CCP
problem, as it is less general than DCCCP.)

\paragraph{Model and problem statement}
Recall that FP is the class of function problems that can be solved by a
deterministic polynomial-time Turing machine, and that
\#P is the class of
counting problems that can be expressed as the number of accepting runs of a
nondeterministic polynomial-time Turing machine. When seeing the output of FP problems as
an integer, we have $FP \subseteq \#$P.
Further, $\text{FP}^{\#\text{P}}$ is the class
of function problems that can be solved in deterministic polynomial time with
access to a \#P oracle: we use this class simply because the answer to our
problems are formally probabilities rather than integers.
In the data complexity perspective, it is easy to see that
the probabilistic membership problem $\probone{L}$ is in
$\text{FP}^{\#\text{P}}$ for any language $L$ which is in the class PTIME,
i.e., where there is a deterministic polynomial-time algorithm testing membership of an input
(non-probabilistic) word to the language $L$: see, e.g.,
\cite[Theorem 4.2]{gradel1998complexity} or \cite[p19]{monet2018combined}. 
All languages studied in the sequel will be PTIME languages in this sense (as is the
case of all CFLs).

We will be interested in languages $L$ for which $\probone{L}$ is in FP (we will
often abuse terminology and say in that case that $\probone{L}$ is in PTIME), and languages
$L$ for which $\probone{L}$ is
\emph{\#P-hard}, i.e., any problem in \#P can be reduced by Turing
reductions to $\probone{L}$: note that we always use Turing reductions (and not
many-one reductions).
When giving complexity bounds with an explicit polynomial
degree, to avoid specifying the details of the computational model we will state
them \emph{up to the cost of arithmetic operations}, i.e., assuming that all
arithmetic operations take unit time. 

Our main goal in this paper is to determine for which languages $L$ the problem $\probone{L}$ can
be solved in polynomial time data complexity, and for which it is
\#P-hard.

\section{CFLs and Unambiguity}
\label{sec:ambiguity}
To start our study of tractable cases of the $\prob$ problem for CFLs, 
we study the effect of \emph{unambiguity}.
On the one hand, we show that $\prob$ is tractable for
\emph{uCFLs}, in data complexity and even in combined
complexity. This implies in particular that $\prob$ is also tractable
for all regular languages, and gives a weighted analogue to the tractability result
of~\cite{bertoni1991complexity} (with similar techniques).

On the other hand, we show that $\prob$
can be \#P-hard for inherently ambiguous CFLs: we give an explicit hard linear CFL. We further show that $\prob$ is already \#P-hard for
unions of two linear uCFLs, though their language is more complicated.
This extends the intractability results
of~\cite{bertoni1991complexity} to the weighted setting, noting that their work instead shows
\emph{\#P$_1$-completeness}~\cite{hardishard} 
because their input instances are unweighted hence fully defined by their
length. Our techniques are similar but we encode Turing machine
computations with the encoding of~\cite{baker1974reversal} (whereas
\cite{bertoni1991complexity} uses \cite{hartmanis1967context}), so our result
applies to linear CFGs.

\paragraph{Unambiguous CFLs}
We focus on the case of unambiguous CFLs and show:

\begin{proposition}
  \label{prp:cyk}
  Given $\Gamma$ a uCFG and $p$ a probabilistic word, we can solve $\prob$ in 
  $O(|p|^3 |\Gamma|)$, assuming unit cost for arithmetic operations.
\end{proposition}

We do not claim that the cubic dependency in the probabilistic word is optimal,
and leave to future work the investigation of whether the complexity can be
lowered, e.g., to match the complexity of Valiant's
parser~\cite{valiant1975general}, or the quadratic upper bound given by
the Earley parser for unambiguous CFGs~\cite{earley1970efficient}.

This result implies that $\prob$ is tractable for many CFGs, e.g., the language
of palindromes, the language $\{a^n b^n \mid n \in \NN\}$, or the language of
Dyck words. It also
implies the tractability in data complexity of subclasses of uCFGs,
for instance \emph{deterministic CFGs}~\cite{ginsburg1965deterministic}, or indeed all \emph{regular
languages}. Our result generalizes the tractability in combined complexity of
(weighted or unweighted) counting for regular languages given as
\emph{unambiguous automata}, see, e.g.,~\cite{colcombet2015unambiguity}.
We note that, by contrast, hardness holds in combined complexity with
nondeterministic automata as input: we come back to this in \cref{sec:morecombined}.

\cref{prp:cyk} in fact will follow immediately from results given later in the
  paper, i.e., by combining \cref{prp:circuittract} and 
  \cref{prp:cykcircuit}. Still, we provide a short alternative proof sketch here:

  \begin{proof}[Proof sketch of \cref{prp:cyk}]
  We first preprocess the input uCFG $\Gamma$ in linear time to make it into \emph{arity-2 normal form}
  (2NF)~\cite{lange2009cnf}. In other words, $P$ consists of productions of the
  following form for various nonterminals $X$:
  \begin{itemize}
    \item $X \rightarrow \epsilon$
    \item $X \rightarrow a$ for $a \in \Sigma$
    \item $X \rightarrow Y$ for $Y \in N$
    \item $X \rightarrow YZ$ for $Y, Z \in N$
  \end{itemize}
  It is known that this operation preserves
  unambiguity~\cite[Appendix~A.2]{amarilli2022efficient}. We then generalize the
  textbook CYK parsing algorithm that computes inductively,
  for each nonterminal $X$ and each interval $1 \leq i \leq j \leq n$, 
  whether the factor of the input word with endpoints $[i,j]$ can be derived
  from~$X$. Instead, we compute the \emph{total probability}
  of the words that can be derived from~$X$ in the distribution represented by
  the factor $(p_i, \ldots p_j)$. As $\Gamma$ is unambiguous, every word admits
  at most one derivation tree, so there are no double counts.
\end{proof}

However, for arbitrary CFGs, we
cannot solve $\prob$ simply by counting the probability of derivation trees in this
fashion -- because the probability of words with multiple derivation trees is counted
as many times as there are trees\footnote{However, the algorithm as presented could
  apply to those CFGs in which all accepted words have precisely
  the same number of derivation trees. We do not know if this strictly
generalizes uCFLs: see~\cite{cstheoryambig}.}.

\paragraph{An intractable CFL}
The previous result shows that $\prob$ is tractable for uCFLs, leaving open the
complexity of inherently ambiguous CFLs. We will
now show intractability in this case by first giving a simple explicit CFL for
which $\prob$ is hard, before strengthening the result to a 2-ambiguous CFL in
the sequel.

Our \#P-hardness  result is by reducing from the problem of counting satisfying
assignments to a certain class of Boolean formulas, called \emph{PP2DNFs}, and
defined below:

\begin{definition}[\cite{suciu2011probabilistic}]
  A \emph{positive partitioned 2-DNF formula}, for short \emph{PP2DNF},
    is an expression on Boolean variables $V = \{x_1, \ldots, x_n, y_1, \ldots, y_n\}$ of the form
    $\bigvee_{1 \leq k \leq m}(x_{i_k} \land y_{j_k})$ where $i_1, \ldots, i_m,
    j_1, \ldots, j_m \in \{1, \ldots, n\}$.
  Given a PP2DNF formula $\Phi$ on variables $V$, a \emph{valuation} $v$
  is a mapping $\nu$ from $V$ to $\{0,1\}$. 
  It is said to be a \emph{satisfying valuation} if, when substituting the
  variables of~$\Phi$ according to~$\nu$, the formula $\Phi$ evaluates to true.
\end{definition}

Satisfiability for PP2DNFs is trivial as there are no negations; by contrast:

\begin{theorem}[\cite{hardishard,suciu2011probabilistic}]
  \label{thm:pp2dnfhard}
  The following problem, written \#PP2DNF, is \#P-hard: given a PP2DNF formula $\Phi$, count the number
  of satisfying valuations of~$\Phi$.
\end{theorem}

We will now define our hard CFL $L_0$ on alphabet $\Sigma \coloneq \{0, 1,
\#\}$:
\[
  L_0 \coloneq \{u \# \Sigma^* v^R \Sigma^* \mid u, v \in
  \{0,1\}^*, v \leq u\}
\]
where we abuse notation and write $\Sigma^*$ to denote arbitrary factors
in~$\Sigma^*$, and where we denote by $\leq$ the order relation defined on $\{0,1\}$ by $0 \leq 1$ 
and extended in a pointwise manner on binary words of the same
length, i.e., we have $v \leq u$ if $|u| = |v|$ and for each $1 \leq i
\leq |u|$ we have $v_i \leq u_i$. Equivalently,
for each $1 \leq i \leq |u|$ such that $v_i = 1$ we also have
$u_i = 1$.

It is not hard to see that $L_0$ is a linear CFL, e.g., it can be recognized by
$\Gamma_0 = (\{S, S_1, S_2\},S,P)$ with productions:
\[
  S  \rightarrow S0 ~|~ S1 ~|~ S\# ~|~ S_1, \hfill
  S_1  \rightarrow 0S_10 ~|~ 1S_10 ~|~ 1S_11 ~|~ \#S_2, \hfill
  S_2  \rightarrow 0S_2 ~|~ 1S_2 ~|~ \#S_2 ~|~ \epsilon
\]

Our first hardness result is then:

\begin{proposition}
  \label{prp:hard1}
  The problem $\probone{L_0}$ is \#P-hard.
\end{proposition}
\begin{proof}
  We show that the DCCCP problem (recall \cref{def:dcccp}) for $L_0$ and $\Sigma' = \{0,1\}$ is \#P-complete, hence the
  same is true for $\probone{L_0}$.
  We reduce from \#PP2DNF, which is \#P-hard by \cref{thm:pp2dnfhard}.
  (We remark that the proof does not use the fact that the formula is
  partitioned, nor that it is in 2DNF: the same proof would apply to the problem
  of counting satisfying assignments of positive DNF formulas in general.
  However, the more stringent requirements on PP2DNF will be used in a
  subsequent hardness result, namely, \cref{prp:counthard}.)

  Let $\Phi$ be an input PP2DNF,
  let $V = \{z_1, \ldots, z_{2n}\}$ be the variables of~$\Phi$, and let
  $C_1, \ldots, C_m$ be the clauses. We code $\Phi$ into a
  partial word $w = u \# v_1^R \# \cdots \# v_m^R \#$ defined in the following
  way:
  \begin{itemize}
    \item The partial word $u$ consists of $2n$ wildcards '?'.
      It intuitively codes the choice of
      valuation of the $2n$ Boolean variables, and these are the only
      wildcards in the coding.
    \item For $1 \leq k \leq m$, the word $v_k$ is the (non-partial) word
      of length $2n$ where the $i$-th letter for $1 \leq i \leq 2n$ is $1$
      if variable $z_i$ occurs in clause $C_k$ and $0$ otherwise. It intuitively
      codes the clause $C_k$.
  \end{itemize}
  This construction is clearly in polynomial time.
  Consider the obvious bijection from the Boolean
  valuations $\nu$ of $V$ to the completions $w'$ of~$w$ defined by choosing a
  completion for~$u$ according to~$\nu$.
  We prove that an arbitrary Boolean valuation $\nu$ satisfies $\Phi$ if and only if its image by this bijection is in $L_0$.

  For the forward direction, assume that $\nu$ satisfies $\Phi$. Let $1 \leq k
  \leq m$ be the index of a clause $C_k = z_{i_k} \land z_{j_k}$ that $\nu$ satisfies. Then, we have
  $\nu(z_{i_k}) = \nu(z_{j_k}) = 1$. The word $v_k$ has $1$ at precisely the
  positions $i_k$ and $j_k$ and zeroes otherwise. Then, letting $u_\nu$ be the
  completion of~$u$ corresponding to~$\nu$, we know that $u_\nu$ carries a $1$ at
  those positions where $v_k$ carries a~$1$. Thus, considering the completion $w_\nu$
  of~$w$ corresponding to~$\nu$, we know that $w_\nu$ matches $u_\nu \# \Sigma^* v_k^R
  \Sigma^*$ with $v_k \leq u_\nu$, so $w_\nu$ is in~$L_0$.

  For the converse direction, assume that for some choice of Boolean valuation
  $\nu$ we have that $w_\nu$ is in $L_0$. By considering the position of
  the first \# in~$w$, and the fact that $u$ and the $v_k^R$ all have the same
  length, we know that there must exist $1 \leq k \leq m$ such that $u_\nu \geq
  v_k^R$. This witnesses that $\nu$ satisfies all variables of clause $C_k$,
  hence $\nu$ satisfies~$\Phi$. This concludes the proof.
\end{proof}

We remark that the hardness proof easily adapts to a more symmetric
language:

\begin{proposition}
  \label{prp:hard1b}
  The problem $\probone{L_0'}$ is \#P-hard for the following CFL on
  $\Sigma = \{0, 1, \#\}$:
  \[L_0' \coloneq
    \{\Sigma^* \# u \# \Sigma^* \# v^R \# \Sigma^* \mid u, v \in \{0,1\}^*, u \leq
v\}.\]
\end{proposition}
\begin{proof}
  The proof is the same as \cref{prp:hard1b} except that the partial word is
  defined, instead of $w = u \# v_1^R \# \cdots \# v_m^R \#$, as:
  \[
    w = \# 0 u 1 \# \# 1 v_1^R 0 \# \cdots 1 v_m^R 0 \#
  \]
  There is still an obvious bijection from the Boolean valuations $\nu$ of the
  variables $V$ to the completions of~$w$. If the $i$-th clause is satisfied
  by~$\nu$, then letting $u'$ be the outcome of~$u$ corresponding to~$\nu$, the
  factors $\# 0 u 1 \# $ and $\# 1 v_i^R 0 \#$ witness that $w' \in L_0'$.
  Conversely, if $w' \in L_0'$, our addition of $0$ and $1$'s ensure that the
  only possible choices are $\# 0 u' 1 \#$ on the one hand, and some $\# 1 v_i^R
  0 \#$ on the other hand, so we conclude like in the original proof that the
  Boolean valuation corresponding to~$u'$ must be a satisfying valuation.
\end{proof}

\paragraph{Intractability for 2-ambiguous linear CFLs}
We now move to our second hardness result:

\begin{proposition}
  \label{prp:hard2}
  Let $\Sigma = \{0,1,\#\}$ and $\Sigma' = \{0,1\}$.
  There are two 
  linear uCFLs $L_\LL$ and $L_\RR$ on~$\Sigma$
  such that
  $\probone{L_\LL \cup L_\RR}$, 
  is \#P-complete.
\end{proposition}

The result above implies in particular\footnote{We do not know whether this is a
strengthening, because it is not known whether
all 2-ambiguous CFLs can be written as the union of two uCFLs, cf 
\cite[Section 9.3, point (iv)]{wich2005ambiguity}.} that $\prob$ is
hard already for 2-ambiguous linear CFLs. 
We note that the situation is very different for regular languages: in data
complexity this is obvious because probabilistic membership to regular languages
is always tractable (via determinization and \cref{prp:cyk}), and in combined
complexity we will show tractability for $k$-ambiguous automata (see
\cref{sec:morecombined}).

We now present the proof of this result in the remainder of this section. We
will reduce directly from the problem of counting the accepting runs of a
nondeterministic polynomial-time Turing machine, which we will abstract away as
a counting problem. The construction overall resembles the proof of 
\cite{bertoni1991complexity} for unweighted counting, but using an encoding of
Turing machine runs similar to the one used in~\cite{baker1974reversal}.

More precisely, let $\Sigma \coloneq \{0,1\}$, a \emph{run}
over $\Sigma$ is a sequence of words $\rho = u_1, \ldots, u_n$ that all have same
length; its \emph{duration} is the integer~$|\rho| \coloneq n$, and its \emph{initial word}
is~$u_1$. Let $\Lambda \subseteq (\Sigma \times \Sigma)^*$ be a language of
\emph{transitions}. A \emph{legal run} $\rho = u_1, \ldots, u_n$ for $\Lambda$
is a run which satisfies the transitions of~$\Lambda$, i.e., 
for each $1 \leq i < n$ we have $u_i
\bowtie u_{i+1} \in \Lambda$, where we define $x \bowtie y$ for any two words having
same length as the word on $\Sigma \times \Sigma$ defined by zipping
together $x$ and $y$ (formally the $k$-th letter is $(a, b)$ for $a$ and
$b$ the respective $k$-th letters of $x$ and~$y$).

The \emph{run counting problem} for $\Sigma$ and
a language $\Lambda \subseteq (\Sigma \times \Sigma)^*$, is the
following. We are given as input a word $u_1 \in \Sigma^*$ of even length, and we must compute
the number of distinct legal runs that start with $u_1$ and have duration $|u_1|$. We claim the
following, which is rather straightforward by choosing $\Lambda$ to code the
transitions of a nondeterministic polynomial-time Turing machine that halts (has
no possible transitions) when the run is rejecting:

\begin{claim}
  \label{clm:turing}
  Let $\Sigma$ be the alphabet $\{0,1\}$.
  There is a regular language 
  $\Lambda \subseteq (\Sigma \times \Sigma)^*$,
  such that the run counting problem for $\Sigma, \Lambda$ is
  \#P-complete.
\end{claim}

\begin{proof}
  Take any \#P-complete counting problem $\Pi$, which by definition can be
  expressed as the number of accepting runs of some nondeterministic polynomial-time Turing
  machine~$\mathcal{M}$. We use standard notions of Turing machines, omitting
  formal definitions: our machine have a single tape that consists initially of
  the input to the problem followed by padding and which is also used as the
  working tape.
  Let $P$ be the polynomial that bounds the maximal running time
  of~$\mathcal{M}$ as a function of the input length: up to multiplying by two
  we can ensure that $P$ takes only even values, and up to increasing $P$ we
  ensure that on an input word $u \in \{0,1\}^n$ of length~$n$ every run of
  $\mathcal{M}$ terminates in \emph{strictly} less than $P(n)$ steps.
  We will assume that the length of the tape of the machine never changes, which
  can be achieved by padding the tape. We
  code the tape symbols of $\mathcal{M}$, its state, and
  markers for the position of the head, as fixed-length words on the fixed
  alphabet $\Sigma = \{0,1\}$. We will write the status of the tape and
  machine after each computation step as a word over $\Sigma^*$, which encodes
  both the tape contents, the head position (as a marker close to the point
  of the tape where the head is located), and the state (also part of the
  marker): we call this a \emph{configuration}. By our assumption that the tape
  length never changes, all such words will have the same length.

  We also assume that, when the machine $\mathcal{M}$ wants to reject, it has no
  successor configuration (i.e., no transitions are possible from a rejecting
  configuration); and when the machine $\mathcal{M}$ wants to accept, it loops
  (in a deterministic way) from its current configuration. This means that, on
  an input $u$ of length~$n$, after $P(n)$ steps the number of runs that are still
  active precisely corresponds to the number of accepting runs on~$u$ (because
  all other runs have finished earlier and had no successor configuration when
  they rejected).

  We thus define $\Lambda$ to be the language accepting the words $x \bowtie y$ such that
  applying one step of the machine changes the configuration from~$x$ to~$y$
  (again assuming that the machine never moves outside its tape). This language
  $\Lambda$ can be taken to be a regular language, because the words over $\Sigma$
  denoting valid configurations can be recognized by a regular language, and
  then the change from $x$ to~$y$ is a local change on~$x$ (i.e., we only do a
  local move of the head, a change of the state next to the head, and a change
  of the tape close to the head).

  We now show a polynomial-time reduction from the problem of counting the
  accepting runs of~$\mathcal{M}$ to the run counting problem on $\Sigma, 
  \Lambda$. Let $u_1$ be the input to~$\mathcal{M}$. Let $N \coloneq P(|u_1|)$,
  and pad $u_1$ to $u_1'$ by
  adding $N-|u_1|+1$ copies of the padding symbol to the right, so that $u_1'$
  has length~$N$; we will only
  consider runs of duration at most $N$, so the machine will never have time to reach
  the end of the padding. Note that the length of $u_1'$ is even by our choice
  of~$P$. Now, there is a bijection between the runs of
  $\mathcal{M}$ of length $N$ with initial tape $u_1$ (in the standard
  sense), and the legal runs according to our definition above with initial word
  $u_1'$ and with duration $N = |u_1'|$. This establishes the reduction, and the
  reduction is clearly in polynomial-time, which concludes.
\end{proof}

Hence, to establish \cref{prp:hard2}, it suffices to show the following (recalling that the hardness of DCCCP for~$L$ implies the hardness of~$\prob$):

\begin{proposition}
  Let $\Sigma = \{0,1,\#\}$ and $\Sigma' = \{0,1\}$,
  and let $\Lambda$ as in \cref{clm:turing}.
  Then there are two linear uCFLs $L_\LL$ and $L_\RR$ such that the run counting
  problem for $\Sigma', P, \Lambda$ reduces in polynomial time to
  DCCCP for $L_\LL \cup L_\RR$ and $\Sigma'$.
\end{proposition}

\begin{proof}
  As a preliminary observation we note that, by the inclusion-exclusion formula, on every
  probabilistic word $p$, the answer to DCCCP for $\LL(\Gamma_\LL) \cup
  \LL(\Gamma_\RR)$ is
  equal to the answer for $\LL(\Gamma_\LL)$ plus the answer for $\LL(\Gamma_\RR)$ minus
  the answer for $\LL(\Gamma_\LL) \cap \LL(\Gamma_\RR)$. As $\LL(\Gamma_\LL)$
  and $\LL(\Gamma_\RR)$
  are uCFLs, by \cref{prp:cyk} we can solve $\prob$ for them in PTIME, hence
  DCCCP.
  Thus, it suffices to show the hardness of DCCCP for $\LL(\Gamma_\LL) \cap
  \LL(\Gamma_\RR)$. 

  We now explain the reduction. 
Intuitively,
runs will be encoded as $w_1 \# w_3 \# \cdots \# w_4^R \#
w_2^R$ as in~\cite{baker1974reversal}, with $L_\LL$ checking the even
transitions and $L_\RR$ checking the odd transitions.
  Remember that the input to the run counting problem for $\Sigma', \Lambda$
  is a word $u_1 \in \{0, 1\}^*$ of even length $n \coloneq |u_1|$; we write $n
  = 2N$. Build a partial word as follows:
  \[
    p = u_1 \# u_3 \# \cdots \# u_{2N-1} \# \# u_{2N} \# \cdots \# u_2
  \]
  where there are $2N$ blocks in total, and where $u_i$ for $i>1$ is simply $n$
  copies of the wildcard $?$.
  The completions of $p$ must be of the form 
  \[
    q = v_1 \# v_3 \# \cdots \# v_{2N-1} \# \# v_{2N}^R \# \cdots \# v_2^R
  \]
  with $v_1 = u_1$ and each $v_i$ is a word of $\{0,1\}^n$. 
  There is an obvious bijection between the outcomes of~$p$ and the runs over
  $\{0,1\}$ starting with the word $u_1$ and having duration $n$.
  Our goal is to
  design the grammars $\Gamma_\LL$ and $\Gamma_\RR$ so as to ensure that $q$ is
  in $\LL(\Gamma_\LL) \cap \LL(\Gamma_\RR)$ if and only if the corresponding run
  $v_1, \ldots, v_{2N}$ is a legal run.

  The linear grammar $\Gamma_\LL$ checks that, for each $1 \leq i \leq N$, we
  have $v_{2i-1} \bowtie v_{2i} \in \Lambda$. Linear rules can describe
  unambiguously the unzipping of words of~$\Lambda$ (following a deterministic
  automaton for the fixed language $\Lambda$). This makes it possible to derive
  protowords of the form \(x\#T\#y^R\) for some non-terminal \(T\) while
  ensuring that \(x\bowtie
  y\) is an element of \(\Lambda\). The deterministic presentation of
  \(\Lambda\) and our conventions (the language \(\Lambda\) accepts the zipping
  of two words having the same length) ensure that the grammar is unambiguous.
  Then the grammar has a rule that produces the empty word immediately after we
  have produced the symbol \(\#\) on each side.

  The linear grammar $\Gamma_\RR$ is similar to \(\Gamma_\LL\). It checks that,
  for each $1 \leq i < N$, we have $v_{2i} \bowtie v_{2i+1} \in \Lambda$. For
  this, we first skip the letters to the left of the word until reaching and
  consuming the leftmost \#. Then, we proceed like $\Gamma_\LL$ but exchanging
  the roles of the left and right. We finish as soon as we reach the
  middle-of-word marker \#\#.

  We have defined $\Gamma_\LL$ and $\Gamma_\RR$, which ensure that an outcome of
  $p$ is in $\LL(\Gamma_\LL) \cap \LL(\Gamma_\RR)$ iff the corresponding run is
  a legal run. Hence, the number of distinct legal runs starting at $u_1$ and of
  duration $|u_1|$ is precisely the answer to DCCCP with $\LL(\Gamma_\LL)
  \cap \LL(\Gamma_\RR)$ and $\Sigma'$ on~$p$, divided by the normalization
  factor (of polynomial size) $2^{(2N-1)n}$. This establishes the correctness of
  the reduction and concludes the proof.
\end{proof}

\section{Poly-Slicewise Unambiguity}
\label{sec:combined}
We have seen in the previous section that $\prob$ is tractable for uCFLs, and
can be intractable already for 2-ambiguous linear CFLs. This leaves open the
question of whether unambiguity is the tractability boundary for $\prob$ on
CFLs. In this section, we show that this is not the case: there are
inherently ambiguous CFLs for which $\prob$ is tractable.
We do so by introducing the notion of \emph{poly-slicewise-unambiguous}
languages, which are immediately seen to generalize uCFLs and to enjoy 
tractable $\prob$.
Then, we state some consequences of this result:
first on \emph{polyslender languages}, then on
\emph{unambiguous polynomial-time counter automata}. Third, we show that, by contrast, $\prob$
can be hard already on very restricted counter automata if we forego the
unambiguity requirement.

\paragraph{Poly-slicewise-unambiguity}
We introduce our more general class of tractable languages:
\begin{definition}
  A language $L$ is \emph{poly-slicewise-unambiguous} if there exists a
  polynomial-time algorithm $\mathcal{A}_L$ for the following task: given as input an integer
  $n \in \NN$ (written in unary), compute a uCFG $\Gamma_{L,n}$ such that
  $\LL(\Gamma_{L,n}) \cap \Sigma^n = L \cap \Sigma^n$.
\end{definition}

Poly-slicewise-unambiguous languages include languages that are not CFLs, hence
tractability will also hold for such languages.
Note that the complexity of $\mathcal{A}_L$ is measured as a function of $n$ only
(as $L$ is fixed).
Further note that a similar notion 
has been studied in~\cite{sakoda1978nondeterminism,yamakami2025power}
under the name “1U”. These are families of sequences of languages $(L_i)_{i\in \mathbb{N}}$ where each $L_i$ has an unambiguous NFA 
of size $p(i)$ for some fixed polynomial $p$. However, unlike us, the definition
of 1U does not
require the NFAs to be computable in PTIME, i.e., these works use
a non-uniform setting.

Of course, any uCFL $L$ is immediately poly-slicewise-unambiguous (taking
$\mathcal{A}_L$ that always returns an unambiguous grammar~$\Gamma$ for $L$), but we will later show that
some inherently ambiguous CFLs are poly-slicewise-unambiguous.
The following tractability result in data complexity immediately follows from
\cref{prp:cyk}:

\begin{proposition}
  \label{prp:poly}
  Let $L$ be poly-slicewise-unambiguous. Then $\probone{L}$ is in
  polynomial time.
\end{proposition}

Indeed, to solve $\probone{L}$ on an input probabilistic word $p$ of length $n$, we
use $\mathcal{A}_L$ to compute in polynomial time the uCFG $\Gamma_{L,n}$, and then
invoke \cref{prp:cyk}. The overall complexity (up to arithmetics) is in $O(n^{k+3})$ if 
$\mathcal{A}_L$ has running time bounded by $O(n^k)$.

We will now show how \cref{prp:poly} specializes to so-called \emph{bounded
CFLs} and to counter automata, and remark in particular that it strictly
generalizes \cref{prp:cyk}.

\paragraph{Bounded CFLs}
We now study the class of \emph{bounded CFLs}, namely,
those CFLs~$L$ for which there exist words $u_1, \ldots, u_k$ such that $L
\subseteq u_1^* \cdots u_k^*$. This class has been studied for a long
time~\cite{ginsburg1964bounded} and was
equivalently characterized as those CFLs that are \emph{polyslender}, i.e.,
there is $k\in \mathbb{N}$ such that for all $n\in \mathbb{N}$, we have $|L\cap
\Sigma^n|$ is in $O(n^k)$.
It is
immediate that bounded CFLs are polyslender, and the reverse inclusion is known
to hold~\cite{ilie2000characterization}.
The following result is easy to show by a variant of CYK:

\begin{proposition}
  \label{prp:polyslender}
  Given a CFG $\Gamma$ for a bounded CFL and an integer $n \in \NN$, we can
  compute the list of words of $\LL(\Gamma) \cap \Sigma^n$ in time polynomial in
  $\Gamma$ and~$n$.
\end{proposition}
\begin{proof}
  We first transform the input bounded CFG $\Gamma$ to a
  2NF CFG $\Gamma'$ recognizing the same language, as explained in the beginning
  of the proofsketch of \cref{prp:cyk}: then $\Gamma'$ is obviously still
  polyslender. 

  Then, we do the standard step of precomputing which nonterminals $X$ are
  \emph{nullable}, i.e., $X$ can derive $\epsilon$. This can be done in
  linear time
  in $\Gamma'$ by applying the following inference rules:
  \begin{itemize}
    \item if we have a production $X \rightarrow \epsilon$ then $X$ is
      nullable;
    \item if we have a production $X \rightarrow Y$ and $Y$ is nullable then $X$
      is nullable;
    \item if we have a production $X \rightarrow YZ$ and $Y$ and $Z$ are both
      nullable then $X$ is nullable.
  \end{itemize}

  If $n = 0$ then we determine whether $\Gamma'$ derives $\epsilon$ or
  not, simply by checking if the axiom is nullable; so it suffices to concentrate on the case $n > 0$.

  We then do the standard transformation of eliminating epsilon productions.
  More precisely, for every production $X \rightarrow YZ$, if $Y$ is nullable we
  add the production $X \rightarrow Z$, and if $Z$ is nullable we add the
  production $X \rightarrow Y$. Then we remove all productions $X \rightarrow
  \epsilon$. One can check that this produces in linear time an 2NF CFG $\Gamma''$
  where $\Gamma''$ contains no nullable nonterminals, and where 
  we have $\LL(\Gamma'') \setminus \{\epsilon\} = \LL(\Gamma') \setminus
  \{\epsilon\}$ (i.e., the recognized languages are the same possibly up to the
  empty word epsilon).

  Next, we process $\Gamma''$ to make sure that all nonterminals are \emph{useful},
  i.e., there is a derivation tree of~$\Gamma''$ in which they appear. This is
  achieved by first computing in linear time the nonterminals that are
  \emph{reachable} from the axiom (i.e., the axiom is reachable, and if a
  nonterminal $X$ occurs in the right-hand-side of a production from a reachable
  nonterminal then $X$ is reachable) and deleting the non-reachable
  nonterminals. Second, we compute in linear time the nonterminals that derive a
  non-empty language (i.e., if there is a production $X \rightarrow a$ then $X$
  derives a non-empty language, and a nonterminal $X$ has a production where all
  nonterminals in the right-hand-side derive a non-empty language then $X$
  derives a non-empty language) and remove the nonterminals that do not derive a
  non-empty language. It is easy to see that this linear-time process does not
  change $\LL(\Gamma'')$. Further, it is now the case that every nonterminal $X$ of
  $\Gamma''$ is useful: we can build a witnessing derivation tree for~$X$ by
  starting from some path witnessing that $X$ is reachable from the axiom, and
  then we replace every leaf of the resulting tree which is labeled by a
  nonterminal by some derivation tree from that nonterminal, using the fact that
  all nonterminals derive a non-empty language.

  Next, we eliminate all cycles of productions of the form $X_1 \rightarrow X_2,
  \ldots, X_k \rightarrow X_1$. 
  We compute
  from $\Gamma''$ in linear time the directed graph formed of the rules $X
  \rightarrow Y$, and compute in linear time its strongly connected components
  (SCCs). Then, for each SCC, introduce a fresh nonterminal $X$ (making it the axiom
  if the SCC contains the axiom), and replace all occurrences of the
  nonterminals of the SCC with $X$ in all productions. Last, delete all
  productions of the form $Y \rightarrow Y$. The resulting grammar $\Gamma'''$ is
  still computed in linear time from~$\Gamma$, is still in 2NF, recognizes the
  same language as $\Gamma$ up to~$\epsilon$, and now ensures that there is a
  total order $<$ on the nonterminals such that whenever we have a rule of the
  form $X \rightarrow Y$ then $X < Y$.%

  We now do a bottom-up computation, for each nonterminal $X$ and length $1 \leq
  i \leq n$, of the language $L_{X,i}$ defined as follows (*):
  the words of length $i$ that can be derived
  from~$X$. The process closely mimics the construction of the proof of
  \cref{prp:cykcircuit}. The set $L_{X,i}$ is the union, across all productions
  with left-hand-side $X$:
  \begin{itemize}
    \item For productions of the form $X \rightarrow a$ with $a$ a terminal, of
      the single-letter word $a$ if $i=1$ and nothing otherwise;
    \item For productions of the form $X \rightarrow Y$ with $Y$ a nonterminal,
      of the set $L_{Y,i}$;
    \item For productions of the form $X \rightarrow YZ$ with $Y$ and $Z$
      nonterminals, of the concatenations $L_{Y,j} L_{Z,i-j}$ across all $1 \leq j
      < i$.
  \end{itemize}
  Note that the definitions are nonrecursive thanks to the order $<$ on sets
  with the same index $i$ and otherwise to the fact that the definition of sets
  with larger $i$ only uses sets with smaller~$i$. It is immediate to show by
  induction that the sets $L_{X,i}$ satisfy invariant (*), which establishes
  correctness. It is also obvious that, since all nonterminals are useful, the
  computed sets are all of polynomial size (since $\Gamma$ is bounded), so the overall computation is
  polynomial because we do polynomially many operations above on these sets of
  polynomial size.
\end{proof}

Thus, bounded CFLs are poly-slicewise-unambiguous,
even in ``combined complexity'':

\begin{claim}
  Any polyslender CFL $L$ is poly-slicewise-unambiguous, and the algorithm
  $\mathcal{A}_L$ can be made to run in PTIME in the input length and in an input
  CFG representing~$L$.
\end{claim}
\begin{proof}
  Fix any CFG $\Gamma$ that recognizes the language.
  The algorithm $\mathcal{A}_L$ proceeds as follows. First, given the length $n
  \in \NN$, use \cref{prp:polyslender} to compute the explicit list of words
  $w_1, \ldots, w_m$ of
  $L \cap \Sigma^n$ in polynomial time. Then, build from this in linear time the
  CFG $S \rightarrow w_1 | \cdots | w_m$. This CFG recognizes $L \cap \Sigma^n$
  by hypothesis, and it is immediate that it is unambiguous. Further, the entire
  process is in polynomial time in $\Gamma$ and~$n$.
\end{proof}

This implies the following combined tractability result of $\prob$ on bounded
CFLs:

\begin{corollary}
  \label{cor:polyslender}
  The $\prob$ problem for polyslender CFLs is tractable in combined complexity.
\end{corollary}

Importantly, some polyslender CFLs are inherently ambiguous, such as 
$\{a^n b^m c^m \mid n, m \in \NN\} \cup \{a^n b^n
c^m \mid n, m \in \NN\}$. Hence, this result covers some languages not covered
by \cref{prp:cyk}. Also note that the results above adapt to any
language which is ``constructively polyslender'' in the sense of 
\cref{prp:polyslender} (even beyond CFLs).

\paragraph{Unambiguous counter automata}
We now present another language formalism enjoying tractable $\prob$
via \cref{prp:poly}, namely, \emph{polynomial-time counter automata}:

\begin{definition}
A \emph{counter automaton} is a machine $\calA =
(Q, q_0, k, F, \delta)$ where $Q$ is a finite set of \emph{states},
$q_0 \in Q$ is the \emph{initial state}, $k\in \NN$ is the number of \emph{counters}, $F
\subseteq Q \times \ZZ^k$ is the \emph{acceptance relation},
and $\delta\subseteq Q \times \ZZ^k \times \Sigma
\times Q \times \ZZ^k$ is the \emph{transition relation}.
  Note that we do not require that $F$ and $\delta$ be finite; instead they will
  be specified by polynomial-time algorithms, as we explain next.

A \emph{configuration} of $\calA$ is a tuple $(q, \vec v)$ with $q \in Q$ and $\vec v \in
\ZZ^k$, the \emph{initial configuration} is $(q_0, \vec 0)$, and a configuration
is \emph{final} if it is in the relation~$F$. A configuration
$(q', \vec v')$ is a \emph{successor} of~$(q, \vec v)$ for the letter $a \in
\Sigma$ if the tuple $(q, \vec v, a, q', \vec v')$ is in the relation $\delta$.
A \emph{run} of $\calA$ on a word $w = a_1 \cdots a_n$ is a sequence of
configurations $c_0, \ldots, c_n$ with $c_0$ the initial configuration and
$c_{i+1}$ a successor of~$c_i$ for~$a_i$ for each $1 \leq i \leq |w|$. The run is
\emph{accepting} if the last configuration is final. The \emph{language}
$\LL(A)$ is the set of words on which $\calA$ has an accepting run.
We say that $\calA$ is \emph{unambiguous} if it has at most one accepting run on every
word.

We say that a counter automaton is \emph{polynomial-time} if there is a
polynomial $P$ satisfying the following two requirements. First,
for every $n \in \ZZ$ and for every accepting run
$(q_0, \vec 0), (q_1, \vec{v_1}), \ldots, (q_n, \vec{v_n})$ of~$\calA$ on a word $w$
of length~$n$, we have $\|\vec{v_i}\|_\infty \leq P(n)$ for all $1 \leq i \leq n$,
where $\|\vec{v}\|_\infty$ denotes the max of the absolute values.
Second, there is an algorithm that decides whether $F(q,\vec v)$ holds on an
input $(q, \vec v) \in Q \times \ZZ^k$ in time bounded by
$P(\|\vec{v}\|_\infty)$, and likewise there is an algorithm deciding whether
$\delta(q, \vec v, a, q', \vec w)$ holds in time bounded by
$P(\|\vec{v}\|_\infty + \|\vec{w}\|_\infty)$.
\end{definition}

The model described here is very general: it allows arbitrary
polynomial-time computations at each step, and it can freely use the value of
counters during the run (i.e., it is not \emph{blind} or \emph{partially blind}
\cite{greibach1978remarks}).
Thus, it generalizes the model of counter machines
from Fischer et al.~\cite{fischer1968counter} when those machines are
required to be \emph{real-time} (i.e., they read one input symbol at each
transition), or more generally when they operate in polynomial-time (so they can only do a polynomial number
of autonomous transitions between input symbols).
Further, the model generalizes
\emph{Vector Addition Systems with States} or
\emph{VASSes}~\cite{karp1969parallel};
\emph{Parikh automata}~\cite{klaedtke2003monadic} (also called
\emph{$\mathbb{Z}$-VASSes}~\cite{czerwinski2020approach} or \emph{integer
VASSes}~\cite{clemente2017separability});
pushdown automata with a unary stack alphabet aka \emph{one-counter
automata}~\cite{valiant1975deterministic}; and extensions such as \emph{Parikh
one-counter automata}~\cite{cadilhac2025parikh}. 
The main restriction is that our counter automata
must read the word in a one-way fashion and in particular they only see each
input symbol once.

The non-probabilistic membership problem for 
polynomial-time counter automata $A$ is in PTIME, and
for unambiguous such automata 
$\prob$ is in PTIME for similar reasons:

\begin{proposition}
  \label{prp:counter}
  Let $\calA$ be an unambiguous polynomial-time counter
  automata. Then $\LL(\calA)$ is poly-slicewise-unambiguous.
  (Hence, $\probone{\LL(A)}$ can be solved in polynomial time.)
\end{proposition}
\begin{proof}
  The polynomial-time algorithm works as follows. Given the length $n$, we build
  an NFA $A_n$ 
  such that
  $\LL(A) \cap \Sigma^n = \LL(A_n) \cap \Sigma^n$. 
The states of $A_n$ are $\{0,
\ldots, n\} \times \mathcal{C}$ for $\mathcal{C}$ the space of potentially
reachable configurations: their
number is polynomial in~$n$ thanks to the polynomial-time bound on~$A$. The final
states of $A_n$ are the states $(n,c)$ for $c \in \mathcal{C}$ a final configuration, which can
be checked in polynomial time for all configurations, and for each $0 \leq i <
n$ and $c, c' \in \mathcal{C}$ there is a transition
from $(i,c)$ to $(i+1,c')$ whenever $c'$ is a successor of~$c$ for~$a$, which can
again be checked in PTIME.

  Now, as $A$
  is unambiguous, we know that $A_n$ is in fact an unambiguous NFA. Further, by construction,
  $A_n$ accepts precisely the words of $\LL(A)$ that have length~$n$. We
  conclude because it is easy to convert unambiguous NFAs to uCFGs in linear time.
\end{proof}

By \cref{prp:poly}, this immediately implies that $\probone{L}$ is tractable for
any language recognized by an unambiguous polynomial-time counter automata: so
the same is true for languages recognized by unambiguous Parikh
automata, unambiguous VASSes, etc.

To contrast with our previous results, note that unambiguous
polynomial-time counter automata and uCFLs are incomparable.
These automata cannot recognize some uCFLs such as the language of palindromes,
intuitively because this requires exponentially many configurations:
see \cite[Theorem~5]{moore2012automata}.
Conversely, these automata can recognize some languages that are not even CFLs (e.g., 
$\{a^n b^n c^n \mid n \in \NN\}$), and they can also recognize some CFLs such as
$\{u \in \Sigma^* \mid |u|_a
\neq |u|_b \lor |u|_a \neq |u|_c\}$ that are not polyslender
and inherently ambiguous~\cite[Proposition~1.10]{berstel1990context}, so their
tractability for $\prob$ did not follow from our previous results.

\paragraph{Nondeterministic counter automata}
Is the unambiguity requirement in \cref{prp:counter} necessary?
We will now
show that $\prob$ is intractable on nondeterministic counter automata, even
for CFLs accepted by very restricted such automata:

\begin{proposition}
  \label{prp:counthard}
  The problem $\prob$ is \#P-hard for $L_1$ and $L_2$ on $\Sigma = \{a,
  \$, \#\}$, with:
  \begin{align*}
    L_1 & \coloneq \{\Sigma^* \# \# a^i \# \Sigma^* \# a^i \# a^j \# \Sigma^* \# 
      a^j \# \# \Sigma^* \mid i, j > 0\}\\
      L_2 & \coloneq \{\Sigma^* \# \# a^i \# \Sigma^* \# a^{i'} \# a^j \# \Sigma^*
        \# a^{j'} \# \# \Sigma^* \mid i, i', j,j' > 0, i+j = i'+j'\}
      \end{align*}
\end{proposition}

\begin{proof}
  We show a reduction from the \#PP2DNF problem, which is \#P-hard by
  \cref{thm:pp2dnfhard}, to the DCCCP problem for $L_1$ and $\Sigma' = \{\#,
  \$\}$: this implies hardness for $\probone{L_1}$. We explain in the proof
  why the reduction also works for DCCCP with $L_2$ and $\Sigma'$, establishing
  the hardness of $\probone{L_2}$. Note that clearly $L_1 \subseteq L_2$.

  Let $\Phi$ be an input PP2DNF,
  let $X = \{x_1, \ldots, x_n\}$ and $Y = \{y_1, \ldots, y_n\}$ be the variables of~$\Phi$, and let
  $C_1, \ldots, C_m$ be the clauses. Let $N \coloneq n+1$. We code $\Phi$ into a
  partial word $p = u v w$, where:
  \begin{itemize}
    \item $u$ codes the status of the variables of~$X$: it is the concatenation
      of words $u_1, \ldots, u_n$, with $u_i = \$ \# ? a^i \# \$$, where all
      letters are deterministic except the wildcard $?$ which stands either
      for $\#$ or for \$;
    \item $v$ codes the clauses of~$\Phi$: it is the concatenation of
      deterministic words $v_1, \ldots, v_m$, with $v_k = \$ \# a^c \# a^{Nd} \#
      \$$, for $x_c$ and $y_d$ the two variables occurring in clause $C_k$
    \item $w$ codes the status of the variables of~$Y$: it is the concatenation
      of words $w_1, \ldots, w_n$, with $w_i = \$ \# a^{Ni} ? \# \$$;
  \end{itemize}
  There is an obvious bijection between valuations of~$X \uplus Y$ and
  completions of the partial word $p$ where the $?$ are substituted by $\#$ for
  true variables and by $\$$ for false variables. We claim that an outcome $q$ of
  $p$ is in $L_1$ iff it is in $L_2$ iff the corresponding valuation $\nu$ is a
  satisfying valuation.

  For the backward direction, if $\nu$ is a satisfying valuation, let $C_k = x_c
  \land y_d$ be a satisfied clause. By isolating the outcomes of $u_c$, $v_i$,
  and $w_d$, we  know that the outcome $q$ of~$p$ is of the following form:
  \[
    \Sigma^* \$ \# \# a^c \# \$ \Sigma^*
    \$ \# a^c \# a^{Nd} \# \$ \Sigma^*  %
    \$ \# a^{Nd} \# \# \$ \Sigma^*
  \]
  From this it is easily seen that $q \in L_1$, so in particular $q \in L_2$.

  For the forward direction, assume that $q \in L_2$. We will show that this
  implies $q \in L_1$, and that it implies the existence of a satisfying
  valuation. For this, consider the occurrences of $\$$ that are required to
  exist in~$p$. The only place where we have factors of the form $\# a^+ \# a^+
  \#$ is in $v$, as in $u$ and $w$ (and in all their outcomes) there is at most
  three $\#$ between any two successive $\$$ and two of these are consecutive.
  Further, the only place where we have factors of the form $\# \# a^+ \#$ is in
  outcomes of~$u$, and the only place where we have factors of the form $\# a^+
  \# \#$ is in outcomes of~$w$.
  So the fact that $q \in L_2$ must be
  witnessed by some factor $\# a^c \# a^{Nd} \#$ in~$v$, corresponding to the
  choice of a clause $C_i = x_c \land y_d$, and a factor $\# \# a^{c'} \#$ in the
  outcome of~$u$ and a factor $\# a^{Nd'} \# \#$ in the outcome of~$w$.  What is
  more, the outcome of $u$ and $w$ must then ensure that the corresponding $?$'s
  have been substituted by~$\#$'s, so that variables $x_{c'}$ and $y_{d'}$ must
  be true.

  We know from $q \in L_2$ that $c + Nd = c' + Nd'$. Notice now that this
  implies $c = c'$ and $Nd = Nd'$. By construction of~$p$, we have $1 \leq c,\,
  c' \leq n < N$. Thus, reasoning modulo~$N$, from $c + Nd = c' + Nd'$ we must
  have $c = c'$, and hence $d = d'$. Thus $q \in L_1$. In particular we have $d =
  d'$. So we know that the variables $x_c$ and $y_d$ are made true by the
  valuation corresponding to the outcome~$q$, thus the valuation satisfies
  clause $C_i$, and then it is a satisfying valuation. This concludes the proof.
\end{proof}

Both $L_1$ and $L_2$ are CFLs because they can be recognized by
nondeterministic pushdown automata with unary stack alphabet: this is a
restriction of our counter automaton formalism, with only one counter and very
simple transition and acceptance relations. The language $L_2$ can even be
recognized by a (nondeterministic) \emph{1-dimensional Parikh automaton}, i.e.,
a counter automaton with a single blind counter which is only tested at the end
of the run.
Hence, hardness for $\prob$ can hold even for very
restricted ambiguous counter automata.

\section{Tractable Circuits}
\label{sec:circuits}
We have seen that $\prob$ is tractable for any uCFL, and also for any
poly-slicewise-unambiguous language (yielding more tractable CFLs).
In this section, we give a more elaborate tractability approach for $\prob$,
using a notion of \emph{tractable multivalued circuits}, that we will follow in
this section and extend in \cref{sec:circuits}. We will show that this formalism ensures the tractability
of~$\prob$, and that it
subsumes uCFLs and poly-slicewise-unambiguous languages.

The section is structured as follows. We first introduce 
our circuits and show that counting is tractable for them.
Then, we show that tractable circuits generalize the tractability of poly-slicewise-unambiguous
languages, and further show that the generalization is strict.

\paragraph{Tractable circuits}
We will define \emph{(smooth) $\times,\uplus$-circuits} following
\cite{amarilli2024tractable},
which are a multivalued analogue to the
\emph{d-DNNF} circuits used in knowledge
compilation~\cite{darwiche2002knowledge}.
A \emph{circuit} $C$ is a directed acyclic graph
with vertices called \emph{gates}.
An \emph{input} to a gate $g$ is a gate $g'$ having a directed edge to~$g$,
the \emph{fan-in} of~$g$ is the number of inputs that $g$ has, and the
\emph{fan-out} of~$g$ is the number of gates of which $g$ is an input. An
\emph{input gate} is a gate with fan-in zero. We assume that $C$ has a
distinguished gate with fan-out zero, called the \emph{output gate}.

Let $\Sigma$ be the alphabet. A \emph{$\times,\uplus$-circuit} on~$\Sigma$ is a
circuit whose input gates are labeled by pairs
$n:a$ for $n>0$ and $a \in \Sigma$, and whose internal gates are labeled with
$\times$ or $\uplus$ and satisfy the three requirements of \emph{decomposability},
\emph{smoothness}, and 
\emph{disjointness} that we now present.
To state the first two requirements, we need to define inductively the \emph{domain} of
each gate $g$ of~$C$: the domain of an input gate $g$ labeled by $n:a$ is
$\dom(g) \coloneq \{n\}$,
and the domain $\dom(g)$ of an internal gate $g$ is the union of $\dom(g')$ for all inputs
$g'$ of~$g$.
The \emph{domain} $\dom(C)$ of~$C$ is that of its output gate.
We then say that a $\times$-gate $g$ is \emph{decomposable} if the
domains of the input gates of~$g$ are pairwise disjoint, and we require that $C$ is
\emph{decomposable}, i.e., every $\times$-gate of $C$ is. We say that a
$\uplus$-gate $g$ is \emph{smooth} if the domains of the input gates of~$g$ are
all identical, and we require that $C$ is \emph{smooth}, i.e., every $\uplus$-gate is.

To state disjointness, we need to define inductively the
\emph{sets of assignments} $\SS(g)$ captured by each gate $g$ of~$C$, as a set
of partial functions $f$ from $\NN$ to~$\Sigma$ with $\dom(f) = \dom(g)$:
\begin{itemize}
  \item For an input gate $g$ of the form $n:a$, we have $\SS(g) \coloneq \{n
    \mapsto a\}$;
  \item For an internal $\uplus$-gate $g$, letting $g_1, \ldots,
    g_k$ be its input gates, we have $\SS(g) \coloneq \biguplus_i \SS(g_i)$;
  \item For an internal $\times$-gate $g$, letting $g_1, \ldots,
    g_k$ be its input gates, we have $\SS(g) \coloneq \bigtimes_i \SS(g_i)$.
\end{itemize}
The $\times$ operator 
denotes the Cartesian product,
where we identify a tuple of functions $(f_1, \ldots, f_k) \in \bigtimes_i
\SS(g_i)$ with the function $f$ defined like each $f_i$ on $\dom(f_i) =
\dom(g_i)$,
noting that by decomposability these domains are disjoint.
In particular, if $g$ is a $\uplus$-gate with no inputs then $\SS(g) \coloneq
\emptyset$, and if $g$ is a $\times$-gate with no inputs then $\SS(g) \coloneq
\{\emptyset_\Sigma\}$ for $\emptyset_\Sigma$ the function to~$\Sigma$ with empty domain.
We will write $\SS(C)$ to mean $\SS(g_0)$ for $g_0$ the output gate of~$C$.
We then say that an $\uplus$-gate $g$ is \emph{disjoint} if for any two inputs
$g'$ and $g''$ of~$g$, the sets $\SS(g')$ and $\SS(g'')$ are disjoint. We
require that $C$ is \emph{disjoint}, i.e., every $\uplus$-gate is.
Notice then that, with these definitions, and in particular thanks to smoothness, the partial assignments of $S(C)$
are in fact \emph{total} assignments over $\dom(C)$.

Given a probabilistic word $p = (p_1, \ldots, p_n)$, and a circuit $C$ with
domain $\dom(C) = \{1, \ldots, n\}$, the \emph{probability} $p(C)$ of~$C$ under~$p$ is
the sum of the probabilities of the assignments of $\SS(C)$ under~$p$, where the
probability of $f\colon \dom(C) \to \Sigma$ under $p$ is simply $p(w_f)$ for the
word $w_f = f(1) \cdots f(n)$.
The point of $\times,\uplus$-circuits $C$ is that we can tractably compute the
probability of~$C$ under a probabilistic word~$p$ without materializing the
(potentially exponential) set $\SS(C)$ of assignments that $C$ captures.
This is the immediate analogue of the tractability of model counting for
\emph{d-DNNFs}~\cite{darwiche2001tractable}:

\begin{proposition}[\cite{darwiche2001tractable}]
  \label{prp:circuittract}
  Given a $\times,\uplus$-circuit $C$ and 
  a 
  probabilistic word $p = (p_1,
  \ldots, p_n)$, we can compute $p(C)$
  in linear time (assuming unit cost for arithmetic operations).
\end{proposition}
\begin{proof}
  We do a bottom-up induction on the gates $g$ of~$C$ to compute the probability
  of~$\SS(g)$ under~$p$:
  \begin{itemize}
    \item For input gates $g$ labeled $i:a$, the probability is simply $p_i(a)$;
    \item For $\uplus$-gates $g$, the probability of $\SS(g)$ is the sum of the
      $\SS(g_i)$ for all inputs $g_i$ of~$g$. Here, the disjointness requirement
      ensures that 
      the union that defines $\SS(g)$ is a disjoint union. Thus, the
      sum is correct because the partial functions in the $\SS(g_i)$ all
      correspond to mutually exclusive outcomes of the probabilistic word
      (this also uses the smoothness requirement to argue that they have the same domain);
    \item For $\times$-gates $g$, the probability of~$\SS(g)$ is the product of
      the probabilities of the $\SS(g_i)$
      for all inputs $g_i$ of~$g$. Here, the decomposability
      requirement ensures that the $\SS(g_i)$ for the inputs $g_i$ of~$g$ are on
      disjoint domains, so they correspond to disjoint sets of positions of~$p$
      and thus to independent outcomes. \qedhere
  \end{itemize}
\end{proof}

\paragraph{Languages admitting tractable circuits}
For $n \in \NN$, we say that a $\times,\uplus$-circuit \emph{captures the $n$-th slice of} a
language $L$ if $\dom(C) = \{1, \ldots, n\}$ and $\SS(C)$ precisely
corresponds to the words of $L \cap \Sigma^n$, i.e., a function $f \colon \{1,
\ldots, n\} \to \Sigma$ is in $\SS(C)$ precisely if the word $f(1) \cdots f(n)$
is in $L$.
We say that a language \emph{admits tractable circuits} if there is an algorithm~$A$ running in polynomial time in its input $n \in \NN$ which computes a
$\uplus,\times$-circuit $A(n)$ that captures the $n$-th slice of~$L$. 
(Note that this is stronger than requiring the mere existence of
polynomial-sized circuits: we assume that they can also be computed in
polynomial time.)
The following is
then a direct consequence of \cref{prp:circuittract}:

\begin{proposition}
  \label{prp:circuitprob}
  If a language $L$ admits tractable circuits then $\probone{L}$ is in PTIME.
\end{proposition}

Which languages admit tractable circuits? In fact, all 
poly-slicewise-unambiguous languages do, because we can build tractable circuits
for uCFLs:

\begin{proposition}
  \label{prp:cykcircuit}
  Given a uCFL $G$ and an integer $n$, we can build in time $O(|G| n^3)$ a
  $\times,\uplus$-circuit that captures the $n$-th slice of~$\LL(G)$.
\end{proposition}
\begin{proof}
  Let $\Gamma = (N, S, P)$ be a uCFG. As explained in the proof sketch of
  \cref{prp:cyk}, up to
  linear-time preprocessing of the input uCFG, we assume that it is in arity-two
  normal form (2NF)
  \cite{lange2009cnf}. 

We then compute the nullable non-terminals as in the proof of \cref{prp:polyslender}.
  In the case where $n = 0$, the circuit to compute is either the circuit $C$
  with $\SS(C) = \emptyset$ that captures no assignments, or the circuit $C$
  with $\SS(C) = \{\emptyset_\Sigma\}$ that captures only the function
  $\emptyset_\Sigma$ with empty domain, depending on whether $\epsilon \in \LL(\Gamma)$ or not,
  i.e., depending on whether the axiom is nullable. 
  In the first case, we take $C$
  to be a circuit consisting of a $\uplus$-gate with no inputs, and in the
  second case we take $C$ to be a circuit consisting of a $\times$-gate with no
  inputs. So in the sequel we assume $n > 0$.

  We then eliminate epsilon productions as in the proof of \cref{prp:polyslender},
  obtaining in linear time 
  a 2NF CFG $\Gamma'$
  where $\Gamma'$ contains no nullable nonterminals, and where 
  we have $\LL(\Gamma') \setminus \{\epsilon\} = \LL(\Gamma) \setminus
  \{\epsilon\}$ (i.e., the recognized languages are the same possibly up to the
  empty word epsilon). We moreover show that this process ensures that $\Gamma'$ is still an unambiguous CFG. We
  show that this is true by describing a bijection between the derivation trees
  of~$\Gamma$ and those of $\Gamma'$ on any non-empty word~$w$.

  In the forward direction, any derivation tree for $\Gamma$ on~$w$ can be
  translated to a derivation tree for $\Gamma'$ on~$w$ by simply removing all
  nodes of the tree that derive $\epsilon$, and changing the nodes where we
  derived a non-empty word by applying a rule $X \rightarrow YZ$ with one of $Y$
  and $Z$ deriving the empty word, to use instead the rule $X \rightarrow Y$ or
  $X \rightarrow Z$. In the backward direction, any derivation tree for
  $\Gamma'$ on~$w$ can be translated to a derivation tree for~$\Gamma$ by
  replacing the nodes labeled with the new productions $X \rightarrow Y$ or $X
  \rightarrow Z$ of $\Gamma'$ created from a production $X \rightarrow YZ$ of
  $\Gamma$, by inserting a derivation of~$\epsilon$ from the missing nonterminal
  $Y$ or $Z$: as $\Gamma$ was unambiguous there is in particular only one
  derivation tree of~$\epsilon$ from the missing nonterminal. The two
  translations described are inverses of each other, so we have indeed described
  a bijection, establishing that the languages of $\Gamma$ and $\Gamma'$ on
  non-empty words are the same, and that $\Gamma'$ is still an unambiguous CFG.

  Next, we preprocess $\Gamma'$ to make sure that all nonterminals are \emph{useful},
  using the same method as in the proof of \cref{prp:polyslender}.

  We now compute an order on the nonterminals to explain in which order they
  are considered in the algorithm and make sure that this order is acyclic.
  Specifically, we let $<$ be an order on nonterminals that ensures that
  whenever we have the production $X \rightarrow Y$ then $X > Y$. Such an order
  can be obtained following a topological sort of the directed graph induced by
  the productions of the form $X \rightarrow Y$, noting that this directed graph
  must be acyclic because $\Gamma'$ is unambiguous (this was not the case for \cref{prp:polyslender}) and all its nonterminals are
  useful. Indeed, assuming by contradiction that we have a cycle $X_1
  \rightarrow X_2, \ldots, X_k \rightarrow X_1$, taking any derivation tree
  witnessing that $X_1$ is useful, we would obtain a different derivation tree
  for the same word by replacing the node labeled $X_1$ by a path obtained by
  following these productions in a loop, and this would contradict the fact that
  $\Gamma'$ is unambiguous.

  We can now describe the construction of the circuit $C$ on an input length~$n$
  for the unambiguous 2NF CFG $\Gamma'$ with no nullable nonterminals and with
  order $<$, with $\Gamma'$ and $<$ having been constructed in linear time
  from~$\Gamma$. For every $1 \leq i \leq j \leq n$ and every nonterminal $X$ we
  construct a gate $g_{X,i,j}$ with $\dom(g_{X,i,j}) = \{i, \ldots, j\}$ which
  ensures the following invariant (*): $\SS(g_{X,i,j})$ is the set of functions corresponding to the
  words on positions $\{i,\ldots,j\}$ that can be derived from~$X$. This gate is
  an $\uplus$-gate of gates $g_{X,i,j,P}$, across all productions $P$ having $X$
  as its left-hand-side. For each such production:
  \begin{itemize}
    \item If the production is $X \rightarrow a$ with $a$ a terminal, then:
      \begin{itemize}
        \item If $i=j$, then $g_{X,i,j,P}$ is an input gate labeled by $i:a$;
        \item Otherwise, $g_{X,i,j,P}$ is an $\uplus$-gate with no inputs
      \end{itemize}
    \item If the production is $X \rightarrow Y$ with $Y$ a nonterminal, then
      $g_{X,i,j,P}$ is the gate $g_{Y,i,j}$;
    \item If the production is $X \rightarrow YZ$ with $Y$ and $Z$ nonterminals,
      then $g_{X,i,j,P}$ is an $\uplus$-gate of $\times$-gates $g_{X,i,j,P,k}$
      across all $i \leq k < j$, with the two inputs to $g_{X,i,j,P,k}$ being
      $g_{Y,i,k}$ and $g_{Z,k+1,j}$.
  \end{itemize}
  The output gate of~$C$ is $g_{S,1,n}$, for $S$ the axiom.

  The fact that $C$ is acyclic is witnessed by the order $<$
  among those gates having the same positions $i,j$ as indices, and otherwise
  by noting that we only have a path from a gate having $i,j$ as indices to
  a gate having $i',j'$ as indices if the interval $[i,j]$ is included in
  $[i',j']$.

  The decomposability and smoothness of~$C$ are easy to
  observe from the invariant that the domain of a gate having $i,j$ as indices
  is always exactly $\{i,\ldots,j\}$.

  The fact that $C$ indeed captures the $n$-th slice of $\LL(\Gamma')$ follows
  from invariant (*) highlighted above.

  The disjointness of~$C$ follows from the unambiguity of~$\Gamma'$: from
  invariant (*), a violation of disjointness would imply that the same word $W$
  can be derived from the same nonterminal $X$ in two different ways (either via
  two different productions, or with the same production $X \rightarrow YZ$ with
  two different splits between $Y$ and $Z$), all of which imply a violation of
  unambiguity by considering some parse tree of $\Gamma'$ on a word $w'$ in
  which $X$ derives~$w$, and building a different parse tree for~$w'$. This
  concludes the proof.

  We last observe that the construction above runs in time $O(n^3 |\Gamma'|)$
  with $|\Gamma'|$ linear in $|\Gamma|$, which establishes our desired
  running-time bound.
\end{proof}

This implies the following on poly-slicewise-unambiguous languages
(in particular uCFLs):

\begin{corollary}
  \label{cor:poly2circ}
  Any poly-slicewise-unambiguous language has tractable
  circuits.
\end{corollary}

\paragraph{A tractable non-poly-slicewise-unambiguous CFL}
We next show that, even within the class of CFLs, the converse of
\cref{cor:poly2circ} does not hold. Thus,
the tractability of $\prob$ for languages admitting tractable circuits
(\cref{prp:circuitprob}) strictly generalizes the same result for
poly-slicewise-unambiguous languages (\cref{prp:poly}). 
For this, consider the language
$L_3 =
\{(a+b)^k a (a+b)^n a (a+b)^{n-k-2} \mid k, n \in \NN, 0 \leq k < n-1\}$
on
$\Sigma = \{a,b\}$,
which consists of the words of length
$2n$ which contain two occurrences of~$a$ separated by exactly $n$ characters.
This language is a CFL which is inherently
ambiguous \cite[Proposition 5.8]{kimelfeld2025formal}. It is shown in 
\cite[Theorem 1]{mengel2025lower} that this language is not
poly-slicewise-unambiguous: they show a lower bound of $2^{\Omega(n)}$ on the
size of any uCFG accepting $L \cap \Sigma^n$, so in particular there is no
algorithm computing such uCFGs in polynomial time in~$n$. However:

\begin{claim}
  \label{clm:hardcircuits}
  The language $L_3$ admits tractable circuits.
\end{claim}
\begin{proof}
  Let $n \in \NN$ be the input length. We will explain how to compute in time
  $O(n)$ a $\times,\uplus$-circuit $C$ that captures the $n$-th slice of $L_3$.
  If $n$ is odd then the $n$-th slice is empty, which can be captured by a
  circuit $C$ with $\SS(C) = \emptyset$. A slight technicality is that,
  to match our definitions, we
  must ensure that $\dom(C) = \{1, \ldots, n\}$; we can achieve this, e.g.,
  with a $\times$-gate $g$ whose inputs are input gates $i: a$ for all $i \in \{1,
  \ldots, n\}$ and some arbitrary $a \in \Sigma$, along with an $\uplus$-gate
  $g'$ 
  with no inputs. This ensures that $\dom(C) = \dom(g) = \{1, \ldots, n\}$ and
  that $\SS(C) = \SS(g) = \emptyset$ because $\SS(g') = \emptyset$. 
  Hence, in the sequel we assume that $n = 2k$.

  The circuit will contain gates $g_1, \ldots, g_k$ with
  $\dom(g_i) = \{i, \ldots, k\} \cup \{k+i, \ldots, 2k\}$ for each $1 \leq i
  \leq k$, satisfying the following invariant (*): $\SS(g_i)$ for each $i$
  describes the choices of letters on the
  positions of $\dom(g_i)$ that ensure that there are two $a$'s at distance $k$
  from each other within the positions $\dom(g_i)$.

  The base case is that $g_k$ is a $\times$-gate of the inputs $k:a$ and $2k:a$,
  which is obviously correct.

  For the induction, let $1 \leq i < k$. To satisfy invariant (*), we use the
  following characterization: either the $i$-th and $(k+i)$-th letter of the
  word are both $a$'s, or one of them is a $b$ and the word satisfies $g_{i+1}$.
  Note that these two choices are mutually exclusive. So we make $g_i$ an
  $\uplus$-gate of four gates $g_{i,x,y}$ for $x,y \in \{a,b\}$, each
  $g_{i,x,y}$ being a $\times$-gate of the two input gates $i:x$ and $i+k:y$ and
  of the following:
  \begin{itemize}
    \item If $x=b$ or $y=b$ (or both), the gate $g_{i+1}$;
    \item Otherwise, if $x=a$ and $y=a$, no additional gate.
  \end{itemize}

  The output gate of the circuit is $g_1$. The decomposability and smoothness of
  the circuit follows by a straightforward induction, and its correctness and
  disjointness follow easily from invariant (*).
\end{proof}

This implies that $\probone{L_3}$ is
tractable.
Intuitively, 
the proof uses the fact that
tractable circuits are not constrained to reading the
word in left-to-right fashion, or in any predefined order: this is in contrast
with uCFGs and with poly-slicewise-unambiguous languages.
This also resembles formalisms such as \emph{Multiple Context Free Grammars}
(MCFGs)~\cite{SEKI1991191}, which generalize CFGs by
deriving tuples of strings instead of strings: MCFGs can recognize
complex parenthetical structures~\cite{kanazawa2016parenthesis}, languages
that do not satisfy strong forms of pumping lemmas~\cite{kanazawa2014pumping},
or images of regular tree languages by MSO tree-to-string
transductions~\cite{kolb2003MSOtrans_language}.
In fact, since MCFGs
admit a generalization of the CYK algorithm~\cite{SEKI1991191}, we could easily
generalize
\cref{prp:cykcircuit} to unambiguous MCFGs. This would recapture
the tractability of $\probone{L_3}$, and also show tractability for
other languages definable by unambiguous MCFGs.

\section{Tractable Circuits with Complementation}
\label{sec:palindromes}
We have introduced the notion of a language \emph{having tractable circuits},
and showed that this implies the tractability of $\prob$. In this section, we
show how to extend tractable circuits with a \emph{complementation
operator}, intuitively corresponding to Boolean negation, and we show that
languages enjoying such circuits are still tractable. We then apply this
technique to show the tractability of $\prob$ for two languages: the
context-free language $\pal^2$ of the concatenation of two palindromes, and the
language of primitive words. It is known that neither of these languages admits
an unambiguous CFG, so their tractability would not follow from \cref{prp:cyk};
but we do not know whether these languages could be handled with the technique
of poly-slicewise-unambiguity, or with tractable circuits without
complementation.

The section is structured as follows. We first present the notion of circuits with
complementation. Then, we first use such circuits to show the tractability of
$\prob$ on the language of
primitive words, because the proof is
simpler. We then turn to the language $\pal^2$.

\paragraph{Circuits with complement gates}
The motivation for extending $\times,\uplus$-circuits with complement gates is
that complements of tractable languages for $\prob$
are always tractable:

\begin{claim}
  \label{clm:complement}
  For any language $L$,
  the problem 
  $\probone{\Sigma^* \setminus L}$ reduces in PTIME to $\probone{L}$.
\end{claim}

\begin{proof}
  The answer to $\probone{L'}$ on~$p$ is simply $1 - A$, for $A$ the answer
  to~$\probone{L}$ on $p$.
\end{proof}

Thus, all our tractability results on languages immediately extend to their
complements (even though the complement of a CFL, or
indeed of a uCFL, is generally not a CFL~\cite{hibbard1966independence}).

However, beyond the use of negation at the top-level, we could also support
negation as an arbitrary intermediate gate. This motivates
the definition of \emph{$\times,\uplus,\compl$-circuits}:

\begin{definition}
  \label{def:neg}
  A \emph{$\times,\uplus,\compl$-circuit} on an alphabet $\Sigma$ is defined like a
  $\times,\uplus$-circuit in \cref{sec:circuits} but where we additionally allow
  \emph{$\compl$-gates}. Such gates must have only one input, and their
  semantics is defined as follows: for $g$ a $\compl$-gate with input $g'$,
  let $S$ be the set of
  all functions from $\dom(g) = \dom(g')$ 
  to~$\Sigma$. Then $\SS(g) \coloneq S \setminus \SS(g')$.
\end{definition}

These gates correspond to negation gates for Boolean circuits, and 
$\times,\uplus,\compl$-circuits can be seen as a multivalued analogue of
the (smooth) \emph{d-D} circuits recently studied
in probabilistic databases~\cite{monet2020solving}. We can
immediately generalize \cref{prp:circuittract} to show:

\begin{proposition}
  \label{prp:circuittract2}
  Given a $\times,\uplus,\compl$-circuit $C$, and a probabilistic word $p = (p_1,
  \ldots, p_n)$, we can compute $p(C)$
  in linear time (assuming unit cost for arithmetic operations).
\end{proposition}
\begin{proof}
  We do the same inductive proof as in \cref{prp:circuittract} with the
  following extra case: for a $\compl$-gate $g$ with input $g'$, the probability
  for~$g$ is simply $1-\eta$, for $\eta$ the probability of~$g'$.
\end{proof}

Hence, we can use $\times,\uplus,\compl$-circuits to show the tractability
of~$\prob$. We say that such a circuit \emph{captures the $n$-th
slice of a language}~$L$ by the obvious generalization of the definition of
\cref{sec:circuits}, and that $L$ \emph{admits tractable
$\times,\uplus,\compl$-circuits} if there is an algorithm which given a length
$n$ computes in polynomial time in~$n$ a $\times,\uplus,\compl$-circuit that
captures the $n$-th slice of~$L$. By the immediate generalization of
\cref{prp:circuitprob}, this implies tractability for $\prob$.

While we will use complement gates in this section, we do not know whether they
are necessary, i.e., whether $\times,\uplus,\compl$-circuits capture the
tractability of~$\prob$ for more languages than $\times,\uplus$-circuits; we come back to this point in the
conclusion.
We also note that the support of negation in relation with (u)CFLs seems
superficially similar to the notions of
\emph{conjunctive grammars} and
\emph{Boolean grammars}~\cite{okhotin2013conjunctive}, in
particular to unambiguous such
grammars~\cite{okhotin2008unambiguous}. However, the semantics of Boolean
grammars are different, because they allow in particular conjunctions that range
over the entire word (and so are not decomposable). Thus, unambiguous
conjunctive grammars include, e.g., the intersections of two uCFGs;
but $\prob$ can be intractable for such languages
(see the proof of \cref{prp:hard2}). This suggests that the
tractability of unambiguous Boolean grammars (e.g., for parsing) does not extend
to counting, so there is no obvious connection to our results.

We will now show that the language of primitive words admits tractable
$\times,\uplus,\compl$-circuits in the sense above, and will then show the same for the
language $\pal^2$.

\paragraph{Primitive words}
Let us first define the language $L_\prim$ of primitive words.  We fix an
arbitrary alphabet $\Sigma$.
A word $w \in \Sigma^*$ is \emph{composite} if there exist $u \in
\Sigma^*$ and $k \geq 2$ such that $w = u^k$. If $w$ is not composite, then we
say that $w$ is \emph{primitive}, i.e., 
a word is primitive if it is not a proper power of another word.
In particular, the empty word is not primitive. The language $L_\prim$ consists
of all primitive words over the alphabet~$\Sigma$.

We note that, over non-unary alphabets, it is not currently known whether
$L_\prim$ is a CFL or not~\cite{ito2014context}.
However, it is known that $L_\prim$ is not a uCFL~\cite{petersen1994ambiguity,koechlin2022new}, and also that $L_\prim$ is not a linear
CFL~\cite{horvath1995strong}.
We do not know whether
$L_\prim$ is poly-slicewise-unambiguous, or whether it admits tractable circuits
without complementation.
Our goal is to show the tractability of $\prob$ for primitive words (and hence
for composite words, by \cref{clm:complement}). Namely:

\begin{proposition}
  \label{prp:prim}
 The problem $\probone{L_\prim}$ can be solved in data
 complexity $O(n^2 |\Sigma|)$ up to the cost of
 arithmetic operations, with $n$ the length of the input word.
\end{proposition}

We prove this result before turning to the language $\pal^2$, which
will use a similar technique. Our proofs rely on the 
standard notion of the (primitive) root and of the order of a word.
For the uniqueness and well-definedness of these notions, see
\cite[Theorem~2.3.4]{shallit2008second}:

\begin{definition}
  For any word $u \in \Sigma^* \setminus \{\epsilon\}$, the \emph{root} of~$u$ is the unique primitive word
  $w$ such that we have $u = w^d$ for a certain integer $d \geq 1$, called the
  \emph{order} of~$u$. Note that if $w$ is primitive then $w=u$ and $d=1$.
\end{definition}

Our algorithm to solve $\probone{L_\prim}$ will proceed by distinguishing words
according to their order. For this, we introduce two families of languages:

\begin{definition}
  \label{def:lm}
    Let $k > 0$. We write $L_k$ for the set of words over $\Sigma^*$
    of order equal to $k$. In particular, $L_1 = L_\prim$. Note that the $L_k$
    are a partition of $\Sigma^*$.

    We write $M_k$ for the set of non-empty words $u$ over $\Sigma^*$ that can be written
    as $u = v^k$ for some $v \in \Sigma^*$ (not necessarily primitive). Note
    that $M_k = \bigcup_{d \geq 1} L_{dk}$.
\end{definition}

To solve $\probone{L_\prim}$, we first show that 
$M_k$ admits tractable $\times,\uplus$-circuits for all $k \in \NN$:

\begin{lemma}
     \label{lem:pasprimi}
     Given two integers
     $1 \leq k \leq n$, we can compute in time $O(n \times |\Sigma|)$ a
     $\times,\uplus$-circuit that captures the $n$-th slice of~$M_k$.
\end{lemma}
\begin{proof}
  If $n$ is not a multiple of $k$, then $M_k \cap \Sigma^n$ is empty: words of
  length $n$ never belong to~$M_k$ because they can never be written as $v^k$
  for some $v \in \Sigma^*$. Thus, we can simply use a circuit $C$ with $\SS(C)
  = \emptyset$; we can ensure that $\dom(C)$ is suitable as we did in the
  beginning of the proof of \cref{clm:hardcircuits}. Hence, we assume in the
  sequel that $n$ is a multiple of~$k$.

  Otherwise, let $d \coloneq n/k$, which is an integer. A word $v$ of length
  $n$ belongs to $M_k$ precisely when the $j$-th letter between each of the $k$
  blocks of length $d$ are all equal, formally, when we have $v_{j+pd} =
  v_{j+qd}$ for each choice of $1 \leq j \leq d$ and $0 \leq p \leq q \leq k-1$.
  We can build a $\times,\uplus$-circuit in the following way. The output of the
  circuit is a 
  $\times$-gate $g$ with inputs $g_1, \ldots, g_d$ for each $1 \leq j \leq d$. We
  will ensure that $\dom(g_j) = \{j+pd\mid 0 \leq p \leq k-1\}$ which ensures
  that $g$ is decomposable. For each gate $g_j$, we will ensure that
  $\SS(g_j)$ precisely contains the partial
  assignments of $\dom(g_j)$ that ensure that all the positions of $\{j+pd \mid
  0 \leq p \leq k-1\}$ contain the same value: this ensures correctness of the
  circuit. For this, we make $g_j$ a $\uplus$-gate with inputs $g_{j,a}$ for
  each $a \in \Sigma$. We will ensure that $\dom(g_{j,a}) = \dom(g_j)$ (ensuring
  smoothness) and will ensure that $\SS(g_{j,a})$ precisely contains the
  partial assignments of $\dom(g_j)$ where the positions of $\{j+pd \mid
  0 \leq p \leq k-1\}$ all contain~$a$. This is correct and ensures that $g_j$
  is deterministic. Now, $g_{j,a}$ is simply a $\times$-gate of input gates
  $j+pd:a$ across all $0 \leq p \leq k-1$. This is decomposable and satisfies
  the requirement. The overall construction is in time $O(n \times |\Sigma|)$,
  which concludes.
\end{proof}

Now, we can construct circuits capturing $L_\prim = L_1$, and in fact $L_k$ for
each $k$, from the circuits capturing the $M_k$. Namely, we show the following, which allows us to conclude the proof of Proposition~\ref{prp:prim}:

\begin{claim}
  \label{lem:ie}
Assume that we have pairwise disjoint sets $L_1', \ldots, L_n'$ of words of length $n$, and define for
each $1 \leq i \leq n$ the sets $M_i' \coloneq \bigcup_{1 \leq d \leq n/i} L_{di}'$.
Given $\times,\uplus$-circuits $C_i$ computing the \(n\)-th slice of $M_i'$ for
each \(1 \leq i \leq n\), we can build in
time $O(n^2 + \sum_i |C_i|)$ a $\times,\uplus,\compl$-circuit $C'$ computing
the \(n\)-th slice of $L_1'$.
\end{claim}

Before proving the result, we restate for convenience the relationship between
the language families $M_k$ and $L_k$ from the main text. For any $k \geq
1$, we have:
\begin{equation}
  \label{eqn:lfromm}
  M_k = \bigcup_{d \geq 1} L_{dk}.
\end{equation}

In particular $M_1 = \Sigma^*\setminus \{\epsilon\}$.
What is more, for any length $n$ we have:
\begin{equation}
  \label{eqn:lfromm2}
M_k \cap \Sigma^n = \bigcup_{1 \leq d \leq n/k} L_{dk} \cap
\Sigma^n,
\end{equation}
noting that the union can now be taken to be a finite union over $1
\leq d \leq n/k$ because the words of $\Sigma^n$ have order at
most~$n$. Hence, we will want to apply Claim~\ref{lem:ie} with $L_i' = L_i \cap \Sigma^n$ and $M'_i = M_i \cap \Sigma^n$, which will give us a circuit for the $n$-th slice of $L_\prim$.

We will then use the following lemma on the possibility to do
\emph{subset-complementation} operations with $\compl$-gates:

\begin{claim}
  \label{clm:subsetcompl}
  From two $\times,\uplus,\compl$-circuits $C$ and $C'$ with $\dom(C) =
  \dom(C')$ and $\SS(C') \subseteq \SS(C)$, we can add constantly many gates to
  obtain a circuit $C''$ with $\SS(C'') = \SS(C) \setminus \SS(C')$.
\end{claim}

\begin{proof}
  We use the following set-theoretic identity: for all $U \subseteq V$ we have
  $V \setminus U = \overline{\overline{V} \uplus U}$ where the union is a
  disjoint union.

  So we take a complement of the output gate of~$C$ with a $\compl$-gate, we do
  the disjoint union with the output gate of~$C'$, and then take again the
  complement with a $\compl$-gate. This does not affect decomposability, we have
  argued that it does not violate disjointness because the union in the
  paragraph above is disjoint, and smoothness is satisfied by the assumption
  $\dom(C) = \dom(C')$ in the lemma statement.
\end{proof}

We can now prove \cref{lem:ie}:

\begin{proof}[Proof of \cref{lem:ie}]
  We prove by downwards induction on $n \geq i \geq 1$ that we can build a
  circuit computing $L_i'$ for each $i$. (Note that, to ensure the polynomial
  running time, we do not build copies of the sub-circuits at every step, but
  re-use at each step the circuits computed for larger~$i$.)

  For the base case of $i=n$, we clearly have that $M_n' =
  L_n'$, so we can simply use the circuit for $M_n'$.

  For smaller $i$, 
  distinguish the case $d=1$ in the union, we have:
  \[
    M_i' = L_i' \cup \bigcup_{2 \leq d \leq n/i}
    L_{di}'
  \]
  The union with $L_i'$ and the large union is a disjoint union, so this implies:
  \[
    L_i'= M_i'\setminus \bigcup_{2 \leq d \leq n/i}
    L_{di}'
  \]
  where the right operand of the set difference is a subset of~$M_i' $.
  Now, by assumption we have a circuit with output gate $g_i$ that computes (the $n$-th slice of) $M_i'$.
  Further, by induction hypothesis, we
  have already computed circuits that compute $L_{di}'$ for
  each $d\geq 2$. We make a $\uplus$-gate $g_i'$ that takes their union,
  which is disjoint by assumption on the
  sets $L_1',\ldots,L_n'$: we also know that all gates in the union have the same
  domain so that smoothness is satisfied. We now use \cref{clm:subsetcompl} to
  do a subset-complement operation, adding constantly many gates to obtaining a
  circuit whose output gate computes $L_i'$. This establishes the induction.

  Finally, the case $i=1$ concludes the proof. An easy accounting of the
  operations performed shows that the running time obeys the bound claimed in
  the statement.
\end{proof}

It is interesting to note that, in this proof, we use negation in a nested
fashion (i.e., not just at top-level, unlike \cref{clm:complement}).
We can then prove \cref{prp:prim}: 

\begin{proof}[Proof of \cref{prp:prim}]
  Applying \cref{lem:ie} (with $L_i' = L_i \cap \Sigma^n$ and $M_i' = M_i \cap \Sigma^n$) to the circuits of~\cref{lem:pasprimi}
gives us, in time $O(n^2 |\Sigma|)$ in the input length $n$, a
$\times,\uplus,\compl$-circuit capturing the $n$-th slice of~$L_\prim$. This 
concludes the proof of \cref{prp:prim}.
\end{proof}

\paragraph{Concatenations of two palindromes}
We now move on to the study of the CFL $\pal^2$, where $\pal$ is the language of
palindromes
on our alphabet $\Sigma$, i.e., 
\[\pal^2 = \{uw\mid  u \text{ and } w \text{ are palindromes}\}.\]
Our main result is:

\begin{theorem}
  \label{mainofpal}
    The problem  $\prob(\pal^2)$ can be solved in data complexity $O(n^3
    |\Sigma|)$ up to the cost of arithmetic operations.
\end{theorem}

It is known that $\pal^2$ is an inherently ambiguous CFL with infinite ambiguity degree~\cite{crestin1972langage}, so the tractability of
$\probone{\pal^2}$ does not follow from \cref{prp:cyk}.
It is incidentally easy to show that $\pal^2$ is not a linear CFL:

\begin{proposition}
  There is no linear CFG recognizing $\pal^2$ on any nontrivial alphabet.
\end{proposition}
\begin{proof}
This result is probably folklore
  but we could not locate a reference, so we prove it here.
  We do the proof on alphabet $\Sigma = \{a,b\}$, the proof obviously
  generalizes to any larger alphabet.

We assume by contradiction that the language is a linear CFL. We use the pumping
lemma for linear CFLs~\cite[Lemma 6]{horvath2010pumping}, which states the
existence of an integer $N \in \NN$ such that any word of the language with
length at least $N$ can be written as $uvwxy$ such that $uv^rwx^ry$ is in the
language for each $r \in \NN$, such that $|vx|>0$, and such that $|uvxy| \leq
N$. Fix $N$ from the lemma statement, and assume without loss of
generality that $N>0$. Consider the word $z = (a^N b a^N) (a^N bb
a^N)$, which is obviously a concatenation of two palindromes. Considering a
factorization of $z = uvwxy$, as $|uvxy| \leq N$, it must be the case that $u =
a^i$, $v = a^j$, $w = a^p b a^N a^N bb a^q$, $x = a^{j'}$, and $y = a^{i'}$.
Now, take $r \coloneq 3N$ and consider
the word
$u v^r w x^r y = a^{i+3Nj+p}b a^{2N} bba^{q+3Nj'+i'}$,
which is supposed to be
a concatenation of two palindromes. It is clear that the left palindrome cannot
include any $b$ of the $bb$ factor at the right,
and that
the right palindrome cannot include the $b$ factor at the left. 
So it must be the case that the left palindrome includes the left $b$,
and the right palindrome includes the right $bb$. 
Now, as $|vx|>0$, we have
$j>0$ (case 1) or $j'>0$ (case 2). In case 1, we have a contradiction because
the left palindrome must include the left $b$ but there are not enough $a$'s to
the right of that~$b$ to match the~$a$'s to the left of the $b$. In case 2, we
have a contradiction by considering the right palindrome and reasoning
symmetrically. This concludes the proof.
\end{proof}

The language $\pal^2$ has been abundantly studied, e.g., in~\cite[Section
3]{guo2015combinatorics}; its words have been called \emph{palindrome
pairs}~\cite{borchert2015words}, or \emph{symmetric
words}~\cite{brlek2004palindromic,fici2024some}.
The efficient recognition of
$\pal^k$ for arbitrary $k$ (e.g., in linear time) was studied
in~\cite{kosolobov2015pal}; see also~\cite{rubinchik2020palindromic}. 
More relevant to our purposes, the problem of counting the number of words of
$\pal^2$ was studied in~\cite{kemp1982number}, but they study unweighted
counting, depending only on the word length and on the alphabet size. 

We now prove \cref{mainofpal} in the remaining of this section.
First we give some definitions. We fix an alphabet $\Sigma$.
We write \(\pale\)
to denote the set of palindromes in \(\Sigma^+\). 
We define the notion of \emph{decomposition into
two palindromes}:

\begin{definition}
  \label{defiaveclememeterme}
  For $v \in \pal^2$, a \emph{decomposition of $v$ into two palindromes} is a
  choice of words $x, y\in \Sigma^*$ such that $v = xy$ and $x \in \pal$ and $y
  \in \pale$. Note that we make the choice to require that $x$ can be empty, but
  $y$ is non-empty. The \emph{number of decompositions} of a word $v \in \pal^2$
  is the number of choices of $x, y$ that are a decomposition of~$v$ into two
  palindromes. We say that $v$ \emph{decomposes uniquely into two palindromes}
  if its number of decompositions is~1.
\end{definition}

Our proof of \cref{mainofpal} is based on a result from~\cite{crestin69these}
and~\cite{crestin1972langage}:

\begin{toappendix}
  \label{app:palindrome}

  In this appendix, we give a self-contained proof of \cref{lem:wordcomb2}. To
  do this, we first prove some technical properties about
  primitive words and then give an alternate characterization of \(\pal^2\).

  We say that two words $u$ and $v$ are \emph{conjugates} if there exist words
  $x$ and $y$ such that we have $u = xy$ and $v = yx$. It is straightforward to
  notice that conjugacy is an equivalence relation: see
  \cite[Theorem~2.4.1]{shallit2008second}. We start with a technical lemma about
  primitive words that is stated in~Section 1.2.1 of \cite{lothaire2002acow}.

  \begin{lemma}
    \label{lem:primite_conjugate_uniqueness}
    Let $w$ be a primitive word. Assume that we can write $w = uv$ and $w =
    u'v'$ such that $v$ and $v'$ are non-empty words and such that $v'u' = vu$.
    Then we have $u = u'$ and $v = v'$.
  \end{lemma}
  \begin{proof}
    We use the fact that, when two words are conjugates, then one is a primitive
    word if and only if the second one also is: see \cite[Theorem
    2.4.2]{shallit2008second}.
    Thus, since \(w = uv\) is primitive, so is its conjugate \(vu\).

    By way of contradiction, let us assume that we can write $w = u'v'$ with 
    \(v'\neq\epsilon\) such that 
    \(v'u'
    = vu\) but without having $u=u'$ and $v=v'$. Given that we have
    $w=uv=u'v'$, this must imply that we have
    both $u \neq u'$ and $v \neq v'$. Up to exchanging $v$ and $v'$, we can assume without loss of
    generality that \(|v'| < |v|\). Then \(w = u'v' = uv\) implies
    that there is a non-empty word \(t\) such that \(v = tv'\) and \(u' = ut\).
    Moreover \(v'u' = vu\) implies that \(v'ut = tv'u\), i.e., $(v'u)t = t
    (v'u)$. We also know that $v'u$ is non-empty because $v'$ was assumed not
    to be empty. We then use the fact that, when we have $xy=yx$ for two
    nonempty words $x$ and $y$, then $x$ and $y$ must be powers of some
    non-empty word, i.e., $x = z^{l_1}$ and $y = z^{l_2}$ for $l_1, l_2 > 0$:
    see \cite[Theorem 2.3.3]{shallit2008second}. Thus, we deduce
    that there is a non-empty word
    \(z\) and integers \(l_1>0\) and \(l_2>0\) such that \(v'u = z^{l_1}\) and \(t
    = z^{l_2}\).
    Therefore, we have \(vu = v'u' = v'ut = z^{l_1+l_2}\). This contradicts
    the fact \(w=vu\) is
    primitive, and concludes the proof.
  \end{proof}

  We then give an alternative characterization of $\pal^2$:

  \begin{claim}[see, e.g., \cite{guo2015combinatorics}, Proposition~5]
    \label{clm:pal2carac}
    A word is in \(\pal^2\) if and only if it is conjugate with its
    mirror image.
  \end{claim}

  The result is shown in \cite{guo2015combinatorics}, but we give the proof
  below to be self-contained:

  \begin{proof}[Proof of \cref{clm:pal2carac}]
    Indeed, if \(u\) is in \(\pal^2\), then \(u = xy\) with
    \(x,\,y\in\pal\), so that \(x^R = x\) and \(y^R = y\). Now \(u^R = y^Rx^R
    = y x\) and thus \(u^R\) is conjugate to \(u\).

    Suppose now that \(u\) and
    \(u^R\) are conjugate, that is \(u = xy\) and \(u^R = yx\). Since \(u^R =
    y^Rx^R\) we must have \(y=y^R\) and \(x=x^R\) and so \(x,\,y \in \pal\), so
    that \(u\) is in \(\pal^2\).
  \end{proof}

  We can now show \cref{lem:wordcomb2}:
\end{toappendix}

\begin{lemmarep}
    \label{lem:wordcomb2}
    Let \(n>0\)
    and \(w\) be a
    primitive word, if
    \(w^n\in\pal^2\), then for all \(m>0\): 
    \begin{itemize}
    \item \(w^m \in \pal^2\), and
    \item The number of decompositions of \(w^m\) into two palindromes (see
      \cref{defiaveclememeterme}) is \(m\).
    \end{itemize}
    In particular, \(w\) has a unique decomposition into two palindromes \(x\)
    and \(y\). Moreover, the pairs of palindromes \(w^m\) decomposes into are exactly the
    pairs of words \(x(yx)^k\) and \(y(xy)^l\) such that \(m = k+l+1\).
  \end{lemmarep}

  As~\cite{crestin69these} and \cite{crestin1972langage} are in French and not
  available online, we provide a self-contained proof of this result in
  Appendix~\ref{app:palindrome}.

\begin{toappendix}

  \begin{proof}
    Let $w$ be a primitive word and $n>0$ be an integer such that
    \(w^n\) is in \(\pal^2\). By \cref{clm:pal2carac},
    this is equivalent to the
    fact that \((w^n)^R = (w^R)^n\) is conjugate to \(w^n\). 
    By \cite[Theorem~2.4.2]{shallit2008second} applied to $w^n$ and $(w^R)^n$, we deduce that
    \(w\) and \(w^R\) are conjugate.
    By \cref{clm:pal2carac} again, this is equivalent to \(w\) being in \(\pal^2\). Notice that, since \(w\) is
    primitive, so is \(w^R\).

    As $w \in \pal^2$, we can write \(w = uv\) and \(w^R = vu\) with \(u,\,v\in\pal\). We assume without
    loss of generality that \(u\) and \(v\) form a decomposition of \(w\), i.e.,
    that \(v\) is not empty. \Cref{lem:primite_conjugate_uniqueness}
    then implies that \(u\), \(v\) is the unique decomposition of \(w\). Indeed,
    assuming by contradiction that we have another decomposition $w = u'v'$ with
    $v'$ non-empty, we would again deduce $w^R = v'u'$ so $vu = v'u'$, and
    \cref{lem:primite_conjugate_uniqueness} allows us to conclude.

    Now, notice that for every \(k\in \mathbb{N}\), \(v(uv)^k\) is a palindrome.
    Indeed, \((v(uv)^{k})^R = (v^R u^R)^{k}v^R = (vu)^{k}v = v(uv)^{k}\).
    Similarly, \(u(vu)^k\) is a palindrome. Thus for \(m>0\), \(w^m = (uv)^m =
    u(vu)^{k}v(uv)^l\) with \(m = k+l+1\). Since both \(u(vu)^{k}\) and
    \(v(uv)^l\) are palindrome and since there are \(m\) pairs \((k,l)\) such
    that \(m = k+l+1\), and because $v$ is non-empty, we have that \(w^m\) has at least \(m\) decompositions.

    To complete the proof, we need to prove that every decomposition of \(w^m\)
    is a pair of words of the form \(u(vu)^{k}\) \(v(uv)^l\) with \(m = k+l+1\).
    Now assume that there is a decomposition of \(w^m\) with a pair of
    palindromes \(w^ku'\) and \(v'w^l\) with \(u'\neq u\), \(v'\neq v\) and
    \(u'v' = w\). Since \(w^ku' = (u'v')^ku'\) is a palindrome, we have that
    \((u'v')^ku' = u'^R(v'^Ru'^R)^k = (u'^Rv'^R)^ku'^R\), this means that \(u'=
    u'^R\) and \(v' = v'^R\) and that \(w = u'v'\) and \(w^R = v'u'\) with
    \(u',\,v'\in \pal\). Since, as we have seen, \(w\) has a unique
    decomposition into two palindromes, the only possibility for such \(u'\) and
    \(v'\) to exist is that they do not form a decomposition of \(w\), i.e. that
    \(v' = \epsilon\) (since we have \(u',\,v'\in \pal\)). In that case,
    \(\epsilon\), \(u'\) is the unique decomposition of \(w\) and thus \(u =
    \epsilon\), \(v = u'=w\). Therefore, the decomposition \(u(vu)^{k}\),
    \(v(uv)^l\) of \(w^m\) is of the form \(w^{k+1}\), \(w^l\) which is (in that particular
    case where \(u=\epsilon\)) one the \(m\) decompositions we exhibited above.
  \end{proof}
\end{toappendix}

Using this lemma, our proof of \cref{mainofpal} deviates
from~\cite{kemp1982number} (which uses generating functions). Namely, we follow
an inclusion-exclusion-based reasoning which is similar to that of primitive
words but more complicated (we partition words based on their order and on also
the offset of their decomposition into two palindromes). 
We exclude the case of words of length 0 as it is trivial.
Thus, let $n \geq 1$ be a word size, let $d$ be a
divisor of~$n$, and let $0 \leq j < n/d$. We define $L_{n,d,j}$ to be the subset of $\pal^2$ consisting of the words $u = v^d$ of length~$n$ and order~$d$ where the left
component of the unique decomposition of $v$ into two palindromes has
length~$j$. In other words, recall that \cref{lem:wordcomb2} implies that $v$
decomposes uniquely into two palindromes \(x\) and \(y\); we require that $|x| =
j$. We call $j$ the \emph{offset}.

We will then prove the following, which implies \cref{mainofpal} via
\cref{prp:circuittract2}:

\begin{claim}
  \label{clm:palcirc}
  Given an integer $n \geq 1$, we can build in time $O(n^3 |\Sigma|)$ a
  $\times,\uplus,\compl$-circuit that captures the $n$-th slice of~$\pal^2$ over
  alphabet $\Sigma$.
\end{claim}

We prove \cref{clm:palcirc} in the rest of the section. As words of length~$n$ can be partitioned by their order (which is a divisor
of~$n$), and as the words of $\pal^2$ of a given length and order can be
partitioned uniquely (by \cref{lem:wordcomb2}) by the offset, we immediately
have the following for all $n \in \NN$:

\begin{equation}
  \label{eqn:decomp}
  \pal^2 \cap \Sigma^n = \biguplus_{d|n} \biguplus_{0 \leq j < n/d} L_{n,d,j}
\end{equation}

Thanks to this equation, the problem reduces to the computation of tractable
circuits for the $L_{n,d,j}$. To achieve this, similarly to the case of
primitive words, let us instead define
$M_{n,d,j}$ for $n \geq 1$ and $d$ a divisor of~$n$ and $0 \leq j < n/d$ in
the following way: $M_{n,d,j}$ contains precisely those words $u$ of length~$n$ which can be written
as $u = v^d$ with $v$ not necessarily primitive (i.e., the order of~$u$ is a
multiple of~$d$), and we have $v \in \pal^2$ with offset $j$, i.e., we can write $v =
xy$ with $|x| = j$ such that $x \in \pal$ and $y \in \pale$.
Observe that words in $M_{n,d,j}$ are in $\pal^2$, as they are of the form $(xy)^d = (x) (y(xy)^{d-1})$ with $x$ and $y$ in $\pal$.

We claim that we can efficiently construct circuits for the $M_{n,d,j}$,
analogously to \cref{lem:pasprimi}:

\begin{claim}
  \label{clm:base}
  Given integers $n \geq 1$, $d$ a divisor of~$n$, and $0 \leq j < n/d$, we
  can compute in time $O(n \times |\Sigma|)$ a $\times,\uplus$-circuit that
  captures the $n$-th slice of $M_{n,d,j}$.
\end{claim}

\begin{proof}
  The intuition is that a word $u = v^d$ of $M_{n,d,j}$ is precisely defined by
  the left half of the first palindrome $x$ and the left half of the right
  palindrome $y$ in the decomposition $v = xy$. So, consider each index of the
  left half of the first palindrome (there are $l \coloneq \lceil j/2 \rceil$), and each
  index of the left half of the second palindrome (there are $r 
  \coloneq \lceil (n/d-j)/2
  \rceil$). We do a product across these $l+r$ indexes: for each index $\iota$,
  the gate $g_\iota$ will have as its domain the position that are required to
  be equated to the position of the index, and this partitions the positions of
  the word in a decomposable fashion. Now, each $g_\iota$ is a disjunction of
  gates $g_{\iota,a}$, over
  all $a \in \Sigma$, of the letter that we choose to put at this position: this
  is a deterministic disjunction, and it will be smooth by the definition of its
  input gates $g_{\iota,a}$.
  The inputs $g_{\iota,a}$ are decomposable
  products of input gates $\kappa:a$ for each position $\kappa$ associated to
  the index $\iota$: there are $2n/d$ inputs for most indexes $\iota$, except
  possibly $n/d$ if $\iota$ is the index corresponding to the middle of $x$ or
  of $y$ if they have odd length.
  The overall circuit size, and the running time of the construction, is $O(n
  \times |\Sigma|)$.
\end{proof}

We now do the analogue of \cref{lem:ie}, using the following central
relationship between the $M_\bullet$ and $L_\bullet$:

\begin{claim}
  \label{clm:lm}
  For each integer $n \geq 1$, divisor $d$ of~$n$, and offset $0 \leq j < n/d$,
  we have the following, with disjoint union:
  \[
    M_{n,d,j} = \biguplus_{p \geq 1, pd | n} L_{n,pd,j \bmod (n/pd)}
  \]
\end{claim}

\begin{proof}
  We consider the definition of the $M_{n,d,j}$ and partition the words
  according to their order. This order $pd$ must be a multiple of~$d$, and a
  divisor of~$n$. Now, we claim that the words $u = v^d = z^{pd}$ of a given
  order $pd$ in $M_{n,d,j}$ are precisely the words of $L_{n,pd,j \bmod
    (n/pd)}$.

  Take \(z\) a primitive word so that \(u = z^{pd}\) and \(u\in M_{n,d,j}\).
  Then \(u\) decomposes into the two palindromes \(u_1\) and \(u_2\) so that
  \(|u_1| = j\). From \cref{lem:wordcomb2}, we obtain that \(z\) decomposes
  uniquely into two palindromes \(x\) and \(y\), and  there is \(k\) and \(l\) so
  that \(u_1 = x(yx)^k\), \(u_2 = y(xy)^l\) and \(pd=k+l+1\). Now \(|z| = |xy|
  \frac{n}{pd}\), \(|u_1| = |x|+\frac{kn}{pd}\), thus \(|x| = j\bmod (n/pd)\).
  If follows that \(u\in L_{n,pd,j\bmod (n/pd)}\).

  Conversely take \(u \in L_{n,pd,j\bmod (n/pd)}\) and a primitive word \(z\) so that
  \(u = z^{pd}\). \cref{lem:wordcomb2} implies that \(z\) decomposes uniquely
  into the palindromes \(x\) and \(y\) so that \(|x| = j\bmod (n/pd)\). We also
  have that \(|z| = \frac{n}{pd}\). We let \(k\) be so that \(j = |x| +
  \frac{kn}{pd}\). \cref{lem:wordcomb2} implies that \(u\) decomposes into \(u_1
  = x(yx)^k\) and \(u_2 = y(xy)^l\) with \(pd=k+l+1\). Thus \(u\) is in
  \(M_{n,d,j}\).
\end{proof}

We can now conclude by computing a $\times,\uplus,\compl$-circuit for each
$L_{n,d,j}$ with $d$ a divisor of~$n$ and $0 \leq j < n/d$.
Similarly to the proof of \cref{lem:ie}, we do so by induction over
decreasing values of~$d$, dealing for each $d$ with all values of~$j$.
The base case of $L_{n,n,0}$ is immediate as it is
obviously the deterministic union over $a \in \Sigma$ of the decomposable
product over $1 \leq i \leq n$ of the input gates $n:a$ (i.e., palindromes of
the form $a^n$, which are all in $\pal^2$ and are precisely the words of $\pal^2
\cap \Sigma^n$ with order $n$).

For the induction case with a divisor $d<n$ of~$n$ and $0\leq j < n/d$,
we can rewrite the equation of \cref{clm:lm} as follows by
distinguishing the case $p=1$, with
the union being disjoint and the right argument of the relative complement being
included in the left complement:

\[
    L_{n,d,j} = M_{n,d,j} \setminus \biguplus_{p > 1, pd | n} L_{n,pd,j \bmod (n/pd)} 
\]
From the $\times,\uplus$-circuit of \cref{clm:base} for $M_{n,d,j}$, and the previously
computed circuits,
using \cref{clm:subsetcompl} like in \cref{lem:ie}, we obtain a
$\times,\uplus,\compl$-circuit for $L_{n,d,j}$. 
We conclude with \cref{eqn:decomp}.
The complexity of the construction consists of up to $n^2$ invocations of
\cref{clm:base}, yielding $O(n^3 |\Sigma|)$, and up to $n^2$ gates for the
$L_{n,d,j}$ that all may have linearly many inputs, yielding a total of  $O(n^3
|\Sigma|)$.
Thus, we have shown \cref{clm:palcirc} and
hence \cref{mainofpal}.

\section{Combined Complexity}
\label{sec:morecombined}
In this final section of the paper, we change our perspective on the
probabilistic membership problem and study it from the angle of \emph{combined
complexity}, i.e., when the input features both a probabilistic word and a
representation of the target language. Unlike data complexity, the problem then
depends on how the target language is represented. 
In this section, we focus on regular languages: we give basic definitions
of \emph{automata} which we use as representations of such languages, and we
recall why the probabilistic membership problem in regular languages is
intractable in combined complexity when the language is given as a nondeterministic
automaton. We give a self-contained proof of this result, which illustrates in
particular that intractability already holds for languages that amount to a
\emph{partial pattern matching} problem for a union of partial patterns.
We then show that, in fact, intractability even holds for the more restricted
setting where the language asks for partial pattern matching to one single
partial pattern. Last, we turn to tractability results and show that we can
recover tractability in the combined complexity setting if we require the input
automaton to be \emph{$k$-ambiguous} for any fixed constant $k>0$.

\paragraph*{Automata definitions and basic results.}
We focus in this section on
representations of languages as \emph{automata}, specifically
\emph{nondeterministic automata} (NFAs): such an automaton is a tuple
$A = (Q, q_0, F, \delta)$ where $Q$ is a finite set of \emph{states},
$q_0 \in Q$ is the \emph{initial state},
$F \subseteq Q$ are 
the \emph{final states}, and $\delta \subseteq Q \times \Sigma \times Q$
is a \emph{transition relation}. 
If $(q,a,q') \in \delta$, we say that there is an \emph{$a$-transition} from~$q$
to~$q'$.
Given a word $w = a_1 \cdots a_n$, a \emph{run} of~$A$ on~$w$ is a sequence of
states $q_0, \ldots, q_n$ where $q_0$ is the initial state and $(q_{i-1}, a,
q_i) \in \delta$ for each $0 < i \leq n$; the run is \emph{accepting} if $q_n
\in F$. The \emph{language accepted by $A$}, denoted by $\LL(A)$, is the set of
words $w \in \Sigma^*$ such that $A$ has an accepting run on~$w$.
The languages that can be recognized by NFAs are the regular
languages, which are a strict subset of the (linear) context-free languages; and every NFA can
be converted in PTIME to a linear CFG that recognizes the same language.

In data complexity, we already know that the probabilistic membership problem in
a regular language accepted by an NFA $A$ can always be solved in polynomial time: we simply construct
an automaton $A'$ that recognizes the same language as $A$ and which is
\emph{unambiguous} (UFA), i.e., it has at most one accepting run on every word. This can
be performed in particular with the standard determinization construction. We can then conclude by \cref{prp:cyk}, as UFAs can be efficiently
converted to linear uCFGs (see also \cite{atteson1998calculating}). In combined complexity, however, the same is not true.
While \cref{prp:cyk} gives us a combined tractability result when the input is a
UFA, tractability does not hold when the input is an NFA. In fact, the problem
is known to be \#P-complete~\cite{alvarez1993very}, even though it admits an
FPRAS~\cite{arenas2021nfa} (which was later generalized to
CFLs~\cite{meel2026cfg}). The \#P-hardness already holds in the case of
uniform probabilistic words, i.e., for the unweighted
counting problem of~\cite{bertoni1991complexity}:

\begin{proposition}[\cite{alvarez1993very}, Folklore]
  \label{prp:nfacombinedhard}
  The following problem is \#P-hard: given as input an NFA $A$ over alphabet $\Sigma=\{0,1\}$ and a length $n\in \NN$, compute
$|\LL(A) \cap \Sigma^n|$.
\end{proposition}

We will now show that \#P-hardness already holds in very restricted cases of
NFAs, subsuming the above result.

\paragraph*{Partial pattern matching.}
We will study the case of NFAs corresponding to a \emph{partial pattern
matching} task. Recall that the \emph{pattern matching problem} asks
us, given a word $u \in \Sigma^*$ and a pattern $v \in \Sigma^*$, to
decide whether $v$ occurs as a \emph{factor} (or contiguous subword) of~$u$,
i.e., whether there exist $s, t \in \Sigma^*$ such that $u = s v t$. In
other words, pattern matching asks us to decide whether $u$ belongs to the
regular language $\Sigma^* v \Sigma^*$. The \emph{probabilistic pattern
matching} problem asks us to compute, given a probabilistic word $p$ and a word
$v \in \Sigma^*$, the answer to probabilistic matching to the language $L =
\Sigma^* v \Sigma^*$, i.e., the probability $p(\Sigma^* v \Sigma^*)$ of the
outcomes of~$p$ that contain $v$ as a factor. It is not difficult to see that
this problem can be solved in polynomial time even in combined complexity:

\begin{proposition}
The probabilistic pattern matching problem can be solved in polynomial time in combined complexity.
\end{proposition}
\begin{proof}
We use the Knuth-Morris-Pratt algorithm
to compute tractably from the pattern~$v$ an automaton $A_v$ which accepts $\Sigma^* v \Sigma^*$
and is deterministic. This automaton is in particular unambiguous, so we can apply \cref{prp:cyk}.
\end{proof}

Let us now study a
generalization of pattern matching, called \emph{partial pattern matching},
where the pattern $v$ is a partial word:

\begin{definition}
  A \emph{partial pattern} over $\Sigma$ is simply a partial word over $\Sigma$. Let
  $v$ be a partial pattern over $\Sigma$, and let $u \in \Sigma^*$ be a word.
  We
  say that $u$ \emph{matches} $v$ if $u$ contains a factor that is a completion of $v$.
  In the \emph{partial pattern matching problem}, we are given a word $u \in
  \Sigma^*$ and a partial pattern $v$, and we must determine whether $u$
  has a factor which is a completion of~$v$.
  Likewise, in the \emph{completion counting problem for partial patterns} (the analogue of the completion counting problem, see
  \cref{def:dcccp}), we are given a partial word $u$ over~$\Sigma$ and a partial pattern
  $v$, and we must count how many completions of~$u$ contain a factor which is a
  completion of~$v$.
  Last, in the
  \emph{probabilistic partial pattern matching problem}, we are given a
  probabilistic word $p$ over~$\Sigma$ and a partial pattern $v$, and we must
  determine the total probability of the outcomes of~$u$
  having a factor which is a completion of~$v$, i.e., the probabilistic
  membership problem to the language $\Sigma^* L_v \Sigma^*$ where $L_v$ is the
  finite language of the words that are completions of~$v$.

  We further
  generalize the partial pattern matching problems defined above to admit
  multiple partial words $v_1, \ldots, v_m$ as input.
  For instance, for the probabilistic partial pattern matching problem, we must compute the total probability of the 
  outcomes of $u$ that match one of $v_1, \ldots, v_m$, i.e., the probability of those outcomes of $u$ 
  that contain a factor that is a completion of one of $v_1,
  \ldots, v_m$.
\end{definition}

The partial pattern matching problem for patterns $v_1, \ldots, v_m$ can be
solved tractably in combined complexity, because it amounts to
testing membership to the language $L_{v_1, \ldots, v_m}$ of the words containing a factor that is a
completion of one of $v_1, \ldots, v_m$; this language is regular, and one can
easily compute in linear time from the partial words $v_1, \ldots, v_m$ an NFA
that accepts $L_{v_1, \ldots, v_m}$. However, the same is not obvious for the
probabilistic partial pattern matching problem (which amounts to
probabilistic membership to the regular language $L_{v_1, \ldots, v_m}$),
or even for the completion counting problem for partial patterns.

It is however rather easy to show that the completion counting problem for
partial pattern matching is \#P-hard in
combined complexity when we allow an unbounded number of partial patterns as
input. This already gives a self-contained proof of \cref{prp:nfacombinedhard};
we give it here for completeness, but it can also be found, e.g., in
\cite[Section 2.2]{irwin2022complexity}:

\begin{proposition}
  \label{prp:multpattern}
  The completion counting problem for partial patterns is
  \#P-hard, even on the alphabet $\Sigma \coloneq \{0, 1\}$. Hence, the same
  holds for probabilistic partial pattern matching. 
\end{proposition}

\begin{proof}
  We reduce from the problem \#Positive-DNF,
  which is \#P-hard: given as input a
  positive DNF
  formula $\Phi = \bigvee_{i=1}^m T_i$ over variables $X=\{x_1,\ldots,x_n\}$
  (where each term $T_i$ is a conjunction of positive literals over $X$), compute $\#\Phi$,
  the number of 
  assignments $\tau\coloneq X \to \{0,1\}$ that satisfy $\Phi$. 
  The reduction works as follows.
  For a word $w\in \Sigma^n$, denote by $\tau_w$ the corresponding assignment,
  i.e., such that $\tau(x_i) = w_i$ for $1\leq i \leq n$.
  Now, for each term $T_j$, we build a partial pattern $v_j$ of length $n$
  whose completions are precisely the words $w \in \Sigma^n$ such that $\tau_w$
  satisfies term $T_j$. This can be accomplished simply by putting at the $k$-th
  position of $v_j$ for $1 \leq k \leq n$ the letter $1$ if the variable $x_k$
  occurs in the term $T_j$ (necessarily positively), and the wildcard $?$ if it
  does not occur.
  It is clear that the partial patterns $v_1, \ldots, v_m$ are each constructed in
  linear time; and it is then obvious that indeed a word $w \in \Sigma^n$ has a
  factor that is a completion of one of the $v_j$'s if and only if it is itself a
  completion of one of the $v_j$'s, that is, if and only if the corresponding
  valuation satisfies the term $T_j$. %
  Thus, we have shown that computing the
  number of satisfying valuations of~$\Phi$ reduces in PTIME to counting the
  number of completions of the partial word $(?)^n$ that contain a factor which
  is a completion of one of the partial patterns $v_1, \ldots, v_m$, which
  concludes 
  our reduction.
\end{proof}

It is then natural to wonder whether the hardness result shown above still
applies even in the restricted setting where we only allow one single partial pattern as
input; note that the proof of \cref{prp:multpattern} given above crucially uses
an unbounded number of patterns. We can show that this is true with a more
elaborate reduction (which first appeared as~\cite{patterns}):

\begin{proposition}
  \label{prp:pattern}
  The completion counting problem for partial patterns over alphabet $\Sigma \coloneq \{0, 1\}$ is
  \#P-hard even when we are given just one partial pattern as input. 
  That is, the
  following problem is \#P-hard: given two partial words $u$ and $v$ over
  $\Sigma$, count the number of completions of $u$ that contain a factor which
  is a completion of~$v$.
  Hence, the probabilistic partial pattern matching problem is also \#P-hard
  over alphabet $\Sigma = \{0,1\}$ even when we are given just one partial
  pattern as input.
\end{proposition}

Similar results to \cref{prp:pattern} appear in the bioinformatics literature.
For instance \cite[Theorem 3.1]{ma2007complexity} shows the
NP-hardness of the completion counting problem\footnote{In their terminology the partial pattern $v$
is a \emph{spaced seed}, the partial word $u$ is a \emph{region},
and the probability of matching is the \emph{sensitivity}.}, over alphabet $\{0,1\}$, and
even applies with the additional restrictions that 
the input partial word is $(?)^n$, and that the partial pattern $v$ is of the
form $\{1,?\}^*$. We remark that these two
restrictions were obeyed in the proof of \cref{prp:multpattern} above (for a
unbounded union of patterns), but they will
not be obeyed in our proof of \cref{prp:pattern} below: 
the pattern $v$ will use the letters $0$ and $1$ in addition to the wildcard
$?$,
and
the partial word $u$ will use the letters $0$ and $1$ and the wildcard $?$.
Yet, our \cref{prp:pattern} is not subsumed by 
\cite[Theorem 3.1]{ma2007complexity}, because 
they
show NP-hardness and not \#P-hardness. We do not know if \#P-hardness can
be shown under their assumptions, e.g., by an adaptation of their proof.

We also note that other variants of the probabilistic partial pattern matching
problem have been considered for bioinformatics, studying for instance the
average number of match occurrences, or the distribution of
matches~\cite{kleffe1990exact,nicodeme1999motif, lladser2008multiple}, or the
patterns $v$ maximizing the match
probability~\cite{ma2009seed,nicolas2008hardness}. Further, counting problems for
partial words have also been studied more generally in computer
science~\cite{blanchet2007algorithmic,manea2010hard,manea2013hardness}. In
particular, \cite[Theorem~11]{manea2013hardness} shows that it is \#P-complete
to count the words of a given length which are the completions of at least one
factor of a given partial word~$v$; note that this is different from
\cref{prp:pattern} because the factor relationship is reversed, i.e., our problem
makes sense when $v$ is shorter than~$u$, unlike theirs.

Let us now prove \cref{prp:pattern}:

\begin{proof}[Proof of \cref{prp:pattern}]
  We reduce from the \#Positive-DNF problem.
  Let $\Phi= \bigvee_{i=1}^m
  T_i$ be such a DNF, over variables $X=\{x_1,\ldots, x_n\}$. We explain how to
  encode $\Phi$ in PTIME to a partial word $u$ and partial pattern $v$ such
  that $u$ contains precisely $n$ "?"s corresponding to variables and, for each
  completion of these $n$ "$?$"s, the resulting word $u'$ has a factor matching
  $v$ precisely when the corresponding Boolean valuation satisfies $\Phi$. 

To simplify, we will describe the words $u$ and $v$ on the expanded alphabet
$\{\$, 0, 1\}$, also allowing wildcards "$?$" which can be replaced by $0$ or
$1$ only; note that this is the analogue of the domain-constrained completion
counting problem from \cref{def:dcccp}. We will explain at the end how to go back to the alphabet $\{0, 1\}$.

The partial pattern $v$ is defined as $v = \$ v_1 \$ v_2 \cdots \$ v_m \$$,
where $v_j$ for $1 \leq j \leq m$ is a word of length $n$ describing the $j$-th
term: for $1 \leq i \leq n$ we have $v_j[i] = 1$ if variable $x_i$ occurs in
term $j$ (positively, because the DNF is monotone) and $?$ if it does not
occur.

A key property is that, for any Boolean valuation $\nu$, writing it as a word
$w_\nu$ of length $n$ over $\{0,1\}$, then $w_\nu$ matches $v_j$ iff $\nu$
satisfies the $j$-th term.

The partial word $u$ is defined as: $u = (\$ 1^n)^{m-1} \$ (?)^n \$ (1^n
\$)^{m-1}$. Intuitively, the middle $(?)^n$ block will correspond to the
valuations of the variables. We claim that, for each valuation $\nu$ of the
variables of $\Phi$, considering the completion $u'$ of $u$ obtained by
substituting the middle $(?)^n$ by the word $w_\nu$ on $\{0,1\}$ that codes
$\nu$, then $u'$ has a factor matching $v$ precisely when $\nu$ satisfies
$\Phi$.

For the easy forward direction, if $\nu$ satisfies $\Phi$, then it satisfies a
term of $\Phi$, say the $j$-th one with $1 \leq j \leq m$. We then know that
$w_\nu$ matches $v_j$, and we extend this to a match of $v$ by matching the
$\$$'s of $v$ to the $\$$'s "around" the middle part of $u$. Note that the other $v_{j'}$ with $j' \neq j$ are matched too, because they contain only "$?$"s and $1$'s and are mapped to a $1^n$ factor of $u$.

For the converse direction, consider the possible matches of $v$ in a
completion of $u$. In $v$ we have $m+1$ occurrences of $\$$ each separated by
$n$ letters, and in $u$ we have $2m$ occurrences of $\$$ each separated by $n$
other letters. So, looking at the occurrences of $x$, the candidate matches
are:
\begin{itemize}
  \item $v_1$ is matched to the first $1^n$ block, ..., $v_m$ is matched to the
    completion of the $(?)^n$ central block
  \item $v_1$ is matched to the second $1^n$ block, ..., $v_{m-1}$ is matched to
    the completion of the $(?)^n$ central block, $v_m$ is matched to the first
    $1^n$ block after the $(?)^n$
  \item \ldots
  \item $v_1$ is matched to the completion of the $(?)^n$ block, $v_2$ is
    matched to the first $1^n$ block after the $(?)^n$, ..., $v_m$ is matched to
    the $(m-1)$-th (and last) $1^n$ block after the $(?)^n$.
\end{itemize}

For each $1 \leq j \leq m$, the $j$-th candidate match above is indeed a match
precisely when the completion of $(?)^n$ matches $v_j$, i.e., corresponds to a
Boolean valuation that satisfies the $j$-th term. So, if a completion of $u$
has a factor matching $v$, then it corresponds to a Boolean valuation that
satisfies some term, i.e., that satisfies $\Phi$.

Last, to go back to $\Sigma = \{0, 1\}$, it is easy to get rid of the $\$$'s, e.g., by the following substitutions:
\begin{itemize}
  \item  $\$$ is replaced by $111$
  \item  $?$ is replaced by $0?0$ (both in $u$ and $v$)
  \item  $1$ is replaced by $010$ (both in $u$ and $v$)
\end{itemize}
These substitutions ensure that, in completions of $u$ and $v$, the occurrences of three consecutive $1$'s precisely correspond to the occurrences of the $\$$'s, so the previous proof adapts.
\end{proof}

\paragraph{$\bm{k}$-ambiguous NFAs.} We have seen that the probabilistic partial
pattern matching problem is intractable in combined complexity, i.e., given a
probabilistic word $p$ and a partial word~$v$ we cannot tractably compute the
probability of the outcomes of $p$ containing a factor which is a completion
of~$v$. For the probabilistic membership problem, this implies that it is
intractable in combined complexity to compute the probability $p(\Sigma^* L_v
\Sigma^*)$ for $L_v$ the language of words that are completions of~$v$, when
given $p$ and $v$ as input. Thus, probabilistic membership is intractable in
combined complexity even when the input language is given as an NFA which is
required to be of the very restricted form corresponding to the languages $L_v$, that is, a
sequence of states $q_0, \ldots, q_{|v|}$, with $q_0$ initial and $q_{|v|}$
final, with transitions from $q_i$ to $q_{i+1}$ labeled either by one letter or
by all letters of~$\Sigma$, and with self-loops on $q_0$ and $q_{|v|}$ labeled
by all letters of~$\Sigma$.

In this section, we study a different kind of restriction that one
can impose on NFAs. Namely, we study \emph{$k$-ambiguous} NFAs for fixed $k>0$,
which are those NFAs $A$ satisfying the following property: on every word $w \in
\Sigma^*$, the number of accepting runs of~$A$ on~$w$ is at most~$k$. In
particular, the 1-ambiguous NFAs are precisely the UFAs, and $k$-ambiguous NFAs
can be efficiently converted to $k$-ambiguous linear CFGs.

Remember that we have studied $k$-ambiguous linear CFGs in \cref{sec:ambiguity}, and
showed that while probabilistic membership is tractable in combined
complexity for uCFGs (and hence for UFAs), it is \#P-hard (even in data complexity) already for
2-ambiguous linear CFLs. We now show that, in combined complexity and with NFAs,
the situation is very different, as the probabilistic membership problem is
tractable in combined complexity on $k$-ambiguous NFAs for any constant $k>0$.

\begin{proposition}
  \label{prp:knfa}
  For any fixed integer $k > 0$,
  given a probabilistic word $p$ and a $k$-ambiguous NFA $A$, we can compute
  $p(\LL(A))$ in time $O(|A|^{4k} |p|)$.
\end{proposition}

Note that, in this result, we make the assumption that the NFA that we are
given must be $k$-ambiguous.
However, this assumption can be tractably checked, as there are algorithms for any fixed
$k>0$ to detect in polynomial time whether an input NFA is
$k$-ambiguous~\cite{stearns1985equivalence}; by contrast the problem is
PSPACE-complete when $k$ is not fixed~\cite{chan1983finite},
see~\cite{weber1991degree}.

We prove \cref{prp:knfa} in the rest of this section. The result heavily
relies on techniques from \cite{raha2018universality}, itself
inspired from~\cite{stearns1985equivalence}, to decide the \emph{universality
problem} on
$k$-ambiguous NFAs. We note that, interestingly, the algorithm that we give
will compute its answer probability by direct numeric computation and \emph{not} using tractable
circuits: this contrasts will all other tractability results given in the
paper, which can be obtained via circuits. The reason is the apparent need to do
arithmetic computations over the
probabilities manipulated by the algorithm with coefficients given by integers;
we do not
see a set interpretation for these computations. Thus, we leave open the question of whether \cref{prp:knfa} could be in fact shown by computing in polynomial time a tractable circuit and then invoking \cref{prp:circuittract2}.

Let us start our proof of \cref{prp:knfa}. The key idea is to partition the
language $\LL(A)$ of the input NFA $A$ as:
\begin{equation}
  \label{eqn:part}
    \LL(A) = \LL_1(A) \uplus \cdots \uplus \LL_k(A)
  \end{equation}
where $\LL_i(A)$ for $1 \leq i \leq k$ is the language of the words of $\LL(A)$
having \emph{exactly} $i$ accepting runs in $A$. We will show that $p(\LL_i(A))$
can be efficiently computed for each $1 \leq i \leq k$, which suffices to
conclude as we can sum the probabilities.

To this end, we will show that we can construct from $A$ a family of automata
$A_1, \ldots, A_k$, with each $A_i$ accepting $\biguplus_{i \leq j \leq k}
\LL_j(A)$ with a specific number of accepting runs:

\begin{lemma}
  \label{lem:automrun}
  For any fixed $k > 0$,
  given a $k$-NFA $A$ and $1 \leq i \leq k$, we can compute in time 
  $O(A^{2k})$
  an NFA $A_i$ that
  accepts precisely the words having at least $i$ accepting runs in~$A$, i.e.,
  $\LL(A_i) = \biguplus_{i \leq j \leq k} \LL_j(A)$, and ensures that for
  each $i \leq j \leq k$ every word of $\LL_j(A)$ has precisely 
  $j!/(j-i)!$
  accepting runs in~$A_i$.
\end{lemma}

\begin{proof}
  Let $A=(Q,q_0,\delta,F)$ be the input NFA, and $i \in
  \{1,\ldots,k\}$. Intuitively, we use a kind of powerset construction
  restricted to subsets of size at most $i$ in order to track an ordered
  tuple of $i$ different runs of $A$. We first describe the construction of
  $A_i$ and then explain why it satisfies the required properties.

  We write $[i]$ to mean the set $\{1, \ldots, i\}$, whose elements are called
  \emph{positions}: these will intuitively
  refer to the $i$ runs of an ordered $i$-tuples of runs of~$A$. The states of $A_i$ are the
  pairs $(\equiv, \phi)$ where $\equiv$ is an equivalence relation on the
  set $[i]$ of positions and
  $\phi$ is a function from $[i]/\equiv$ to $Q$. Intuitively, the
  equivalence relation $\equiv$ indicates which positions contain runs that
  are still equal, and which positions contain runs that are known to be
  different
  (even though two runs that have been witnessed to be different in the past may
  be currently located in the same state, i.e., the function $\phi$ is not
  necessarily injective). We will call the equivalence classes
  of $\equiv$ a \emph{batch}.

  The initial state of~$A_i$ is $(\equiv_0, \phi_0)$, where $\equiv_0$ is the
  single-batch equivalence relation that identifies all positions, and where
  $\phi_0$ maps the unique batch to the initial state $q_0$ of~$A$.

  The final states of~$A_i$ are of the form $(\equiv_{\mathrm{f}}, \phi)$, where
  $\equiv_{\mathrm{f}}$ is the equivalence relation with each position
  being in its own singleton batch, and where $\phi\colon
  [i]/\equiv_{\mathrm{f}}\to Q$ is a function such that, for each batch $b \in
  [i]/\equiv_{\mathrm{f}}$, the state $\phi(b)$ is in~$F$, i.e., is a final
  state of~$A$.

  The transitions of~$A_i$ are defined in the following way. Let $(\equiv,
  \phi)$ be the current state, and let $a$ be the letter being read. We pick a
  refinement $\equiv'$ of $\equiv$ (that is, an equivalence relation ensuring
  that whenever $x \equiv' y$ then $x \equiv y$), and a function $\phi' \colon
  [i]/\equiv' \to Q$, subject to the following requirements:
  \begin{itemize}
    \item For each position $x \in [i]$, letting $[x]_\equiv$ be its equivalence
      class in $\equiv$ and $[x]_{\equiv'}$ be its equivalence
      class in $\equiv'$, then there is an $a$-transition from
      $\phi([x]_\equiv)$
      to $\phi'([x]_{\equiv'})$ in~$A$. In other words, when looking at
      positions, they move according to the transitions of~$A$.
    \item For any two positions $x \neq y$ such that  $x
      \equiv y$ (noting that this means $[x]_\equiv =[y]_\equiv$, so that
      $\phi([x]_\equiv) = \phi([y]_\equiv)$), we have that $x \equiv' y$ if and only if  
      $\phi([x]_{\equiv'}) = \phi([y]_{\equiv'})$. In other words, two positions of the same batch become distinguished whenever 
     they become 
      mapped to different states, and otherwise they stay in the same batch.
  \end{itemize}
  If these requirements are satisfied, then we have an $a$-transition 
  from~$(\equiv, \phi)$ to $(\equiv', \phi')$ in~$A_i$. Note that, in
  particular, if some state of~$A$ in the image of~$\phi$ has no outgoing
  $a$-transition, then $(\equiv, \phi)$ will not have any outgoing
  $a$-transition.

  An equivalent way to see the definition above for the $a$-transitions 
  going out of $(\equiv, \phi)$ is
  to define them as follows: we define $\equiv'$ by picking for each
  batch $b \in [i]/\equiv$ an equivalence relation $\equiv_b$ on~$b$,
  and we choose pairwise distinct states to define the $\phi'(b')$ for each $b'
  \in b/\equiv_b$ among the states having an $a$-transition
  from~$\phi(b)$.

  Let us bound the complexity of this construction. The number of equivalence
  relations $\equiv$ over $[i]$ is constant and, for each $\equiv$, the number of possible functions $\phi$
  is $|Q|^k$.  For each state $(\equiv, \phi)$, when building outgoing
  $a$-transitions, we consider every refinement of~$\equiv$ (of which
  there are a constant number), we read the outgoing $a$-transitions from the
  images $\phi(b)$ (an
  additive term of $O(|A|)$ when counted across all letters $a$), and we pick some choice of
  function $\phi'$ among $O(|Q|^k)$ possibilities. Hence, the overall complexity
  is indeed bounded by $O(|A|^{2k})$.

  Let us now argue that the construction is correct, i.e., that $A_i$ satisfies
  the requirements. To do this, we will show that for every word $w$ there is a
  bijection $\eta_w$ from the ordered $i$-tuples
  of (partial) runs of~$A$ on~$w$ to the (partial) runs of~$A_i$ on~$w$. We will
  next argue that this bijection preserves acceptance of the runs, which will
  allow us to conclude.

  For every word $w$, we define the function
  $\eta_w$ that maps an ordered $i$-tuple $(\rho_1, \ldots, \rho_i)$ of partial
  runs of $A$ on~$w$ to the sequence $(\equiv_0, \phi_0), \ldots, (\equiv_{|w|},
  \phi_{|w|})$ of states of $A_i$ defined as follows: for each $0
  \leq j \leq |w|$, we let $(\rho_1^j, \ldots, \rho_i^j)$ be the ordered
  $i$-tuple of runs of~$A$ on the prefix of length~$j$ of~$w$ obtained
  with each $\rho^j_x$ being the prefix of $\rho_x$ of length~$j$, we let $\equiv_j$ be
  the equivalence relation on~$x,y\in[i]$ defined by setting $x \equiv_j y$
  iff $\rho_x^j = \rho_y^j$, and we let $\phi_j$ be the function from
  $[i]/\equiv_j$ to $Q$ obtained by mapping each batch $b \in [i]/\equiv_j$ to
  the common final state of the runs $\rho_x^j$ for $x \in b$ (this state is
  common because those runs are equal).

  We first claim that, indeed, the image of $\eta_w$ consists of runs of~$A$
  on~$w$. This is shown inductively by verifying that, by our definitions, there
  is indeed an $a$-transition from $(\equiv_{j-1}, \phi_{j-1})$ to
  $(\equiv_j, \phi_j)$ for every $0 < j \leq |w|$. Indeed, for each batch $b$ of
  $\equiv_{j-1}$, to go from $(\rho_1^{j-1}, \ldots, \rho_i^{j-1})$, we are
  indeed choosing a refinement by determining which of the runs at positions
  of~$b$ will now be distinguished by the next transition that they do,
  amounting to picking different choices of outgoing $a$-transitions
  from $\phi_{j-1}(b)$.

  We next claim that the function $\eta_w$ is indeed a bijection by exhibiting
  an inverse. Consider a run $(\equiv_0, \phi_0), \ldots, (\equiv_{|w|},
  \phi_{|w|})$ of $A_i$, we define an ordered $i$-tuple $(\rho_1, \ldots, \rho_i)$ of
  runs of~$A$ on~$w$ by defining for $x \in [i]$ the $x$-th run
  $\rho_x$ to be sequence of states $\phi_j([x]_{\equiv_j})$ for $0 \leq j \leq
  |w|$.
  It is easy to see that this defines an inverse to
  the function $\eta_w$, i.e., that applying $\eta_w$ to an ordered $i$-tuple
  $(\rho_1, \ldots, \rho_i)$ of runs of~$A$ on~$w$ and then doing the
  operation above recovers the initial tuple $(\rho_1, \ldots, \rho_i)$.

  We last claim that, in particular, for every word~$w$, the bijection $\eta_w$
  satisfies the following: an ordered $i$-tuples of 
  runs of~$A$ on~$w$ consists of runs that are pairwise distinct and all
  accepting if and only if the resulting run of~$A_i$ on~$w$ ends at a
  final state. Indeed, our bijection ensures that the state $(\equiv, \phi)$
  reached in~$A_i$ is final precisely when all runs are
  different and when the images by~$\phi$ of every singleton batch are all
  mapped to accepting states; and indeed $\phi$ maps each batch to the last state of the
  corresponding run.

  We thus know that $A_i$ accepts precisely the words on which $A$ has at least
  $i$ distinct accepting runs, i.e., the words of $\LL_j(A)$ for some $i \leq j
  \leq k$ thanks to the fact that $A$ is a $k$-NFA (so it has at most $k$
  accepting runs on any given word). Further, on a word $w \in \LL_j(A)$, we know
  that $A$ has precisely $j$ accepting runs, which implies that $A_i$ has
  precisely $j!/(j-i)!$ accepting runs. Indeed, an ordered tuple of $i$ pairwise
  distinct accepting runs of~$A$ is obtained by picking one of the $j$ accepting runs
  for the first position, one of the $j-1$ remaining runs for the second
  position, and so on until multiplying by $(i-(j-1))$, i.e., $j!/(j-i)!$
  possibilities. This concludes the proof.
\end{proof}

We will also need the following lemma, intuitively saying that we
can compute the total number of runs of length $n$ of an NFA, each run being
weighted by the probability of its underlying word. The use of this lemma is the
reason why we do not see how to prove the result using circuits:

\begin{lemma}
  \label{lem:nfaproba}
  Given an NFA $A$ and a probabilistic word $p = (p_1, \ldots, p_n)$, we can compute in time 
  $O(|A|^2|p|)$ the following quantity:
  \[
    \pi(A, p) \coloneq \sum_{w \in \Sigma^n} p(w) \times \nruns(A, w)
  \]
  where $\nruns(A, w)$ is the number of runs of $A$ on a word $w \in \Sigma^*$
\end{lemma}

\begin{proof}
  Let $A=(Q,q_0,\delta,F)$ be the input NFA.
  For $0 \leq i \leq n$, write $\restr{p}{i}$ to be the probabilistic word
  $(p_1, \ldots, p_i)$: note that for $i=0$ this is just the empty probabilistic
  word defined by $\restr{p}{0}(\epsilon) = 1$.
  We define the quantity $\pi_{q,i}$ for $0 \leq i \leq n$ and $q \in Q$ as 
  the sum, over $w \in \Sigma^i$, of $\restr{p}{i}(w) \times \nruns(A, w, q)$ 
  where
  $\nruns(A, w, q)$ denotes the number of runs of $A$ that read $w$ while
  going from an initial state to~$q$. It is then clear that the value that we
  wish to compute in the lemma is the sum of the $\pi_{q,n}$ for $q \in F$.

  We will show how to compute the quantities $\pi_{q,i}$ by increasing value
  of~$i$, and simultaneously argue inductively that it is correct.
  Initially, we set $\pi_{q,0}$ to be $1$ if $q=q_0$ and $0$ otherwise;
  this clearly satisfies the invariant. 

  For $1 \leq i \leq n$, we have that $\pi_{q,i} = \sum_{a \in \Sigma} p_i(a)
  \times \left( \sum_{q' \in Q, q' \rightarrow^a q} \pi_{q',i-1} \right)$ where $q'
  \rightarrow^a q$ denotes that there is an $a$-transition from~$q'$ to~$q$. The
  reason why this is correct is that we can split the sum over $w \in \Sigma^i$
  that defines $\pi_{q,i}$ according to the last letter of~$w$, and for each
  such letter $a \in \Sigma$ the runs over words ending with~$a$ can be
  partitioned according to the state that they reach before reading the~$a$.  Now, remembering that we consider arithmetic operations to take unit time, we indeed end up with complexity $O(|A|^2|p|)$.
\end{proof}

With this in place, we are ready to prove \cref{prp:knfa}:

\begin{proof}[Proof of \cref{prp:knfa}]
Given $A$, we first compute the automata $A_1, \ldots, A_k$ given by
\cref{lem:automrun} in time $O(|A|^{2k})$ (remember that $k$ is a fixed constant), each of which is of size $O(|A|^{2k})$. Then,
given the probabilistic word $p = p_1, \ldots, p_n$, for each $1 \leq i \leq k$
we compute $\pi(A_i, p)$ using \cref{lem:nfaproba}. This takes time $O(|A|^{4k} |p|)$. We now use \cref{eqn:part}:
we know by
\cref{lem:automrun} 
that $\nruns(A_i, w)$ for $w \in \LL_j(A)$ is precisely 
$j!/(j-i)!$
if $j \geq i$ and $0$ otherwise. Hence, we know that:
\[
  \pi(A_i, p) = \sum_{i \leq j \leq k} 
  (j!/(j-i)!)
  p(\LL_j(A))
\]
We can write this as a system of linear equations relating the values
$\pi(A_i,p)$, which we have computed, and the values $p(\LL_j(A))$, which we
want to compute. The matrix of this equation system is a triangular matrix 
$M = \left(j!/(j-i)!\right)_{i,j}$
of size $k\times k$, so we can invert it in constant time (remembering again that $k$ is fixed) to recover the values $p(L_i(A))$, and then the value $p(\LL(A)) = \sum_{1\leq i \leq k} p(L_i(A))$. 
This concludes the proof of \cref{prp:knfa}.
\end{proof}

Beyond the open question of proving \cref{prp:knfa} via tractable circuits, we
conjecture that similar results would hold for
the generalization of probabilistic membership in
$k$-ambiguous nondeterministic tree automata running over probabilistic trees
(i.e., with a fixed skeleton and probabilistic node labels); see
also~\cite{seidl1990deciding} which studies the equivalence problem of such
automata. By contrast, the fact that tractability does not hold for 2-ambiguous
linear CFLs seems to be caused by the two possible parse trees
having different shapes.

\section{Conclusion and Future Work}
\label{sec:conclusion}
We have studied the membership problem for probabilistic
words to CFLs~$L$, which generalizes the problems of counting how many words of a given
length are in~$L$, or how many partial word completions are in~$L$. We have
shown that the problem is tractable for unambiguous CFLs, for
poly-slicewise-unambiguous languages, and for languages admitting tractable
circuits with (decomposable) Cartesian product, (deterministic) disjoint union,
and negation. We have shown that the problem is intractable already for unions
of two linear uCFLs, or for some languages recognized by restricted kinds of
counter automata.

We do not give a complete dichotomy between tractable and intractable
CFLs for
$\prob$. This is not so surprising in hindsight:
via the technique of Greibach~\cite{greibach1968note} we can
easily show the undecidability of the \emph{meta-problem}. Namely, given a CFG
$G$, it is conditionally undecidable to determine whether $\probone{\LL(G)}$ is tractable,
in fact already for linear CFGs:

\begin{proposition}
  \label{prp:meta}
  Assume that $\text{FP} \neq \#$P.
  Then the following problem is undecidable:
  given a linear CFG $\Gamma$, determine whether $\probone{\LL(\Gamma)}$ is in
  FP.
\end{proposition}
\begin{proof}
  We show \cref{prp:meta} by reducing from the universality problem
  for linear CFGs. Recall that this is the problem of determining whether 
  an input linear CFG $\Gamma$ on an alphabet $\Sigma$
  is universal, that is, whether $\LL(\Gamma) =
  \Sigma^*$. This problem is shown undecidable in~\cite{baker1974reversal}.

  We now explain our reduction.
Let $\Gamma$ 
  be the input linear CFL
  on an alphabet $\Sigma$, for which we want to decide universality. 
  Let $\Delta$ be a linear CFL for the hard language $L_0$ of \cref{prp:hard1},
  for which $\probone{L_0}$ is \#P-hard, which is on the alphabet $\Theta = \{0,
  1, \#\}$.
  Up to renaming, we assume without loss of generality
  that $\Sigma$ and $\Theta$ are disjoint.
  We now consider the language $L \coloneq \Sigma^* L_0 \cup \LL(\Gamma) \Theta^*$. We
  can easily construct a linear CFG $\Xi$ that recognizes~$L$: we achieve the
  first member of the union by a linear CFG built from~$\Delta$ by generating an
  arbitrary word of~$\Sigma^*$ to the left before generating a word
  of~$\LL(\Delta)$, and do the analogous construction for the second member of the union.

  We now claim that if $\Gamma$ is universal then $L$ is regular so that
  $\probone{L}$ is in FP, and if~$\Gamma$ is not universal then $\probone{L}$ is
  \#P-hard. Thus, under our assumption, this establishes that it is undecidable
  to determine whether $\probone{L}$ is in FP.

  First, if $\Gamma$ is universal, then $L$ simplifies to $\Sigma^* \Theta^*$.
  This is a regular language, hence it is a uCFL, so that $\probone{L}$ is then in FP by
  \cref{prp:cyk}.

  Second, if $\Gamma$ is not universal, we will show  how $\probone{L_0}$
  reduces in PTIME to $\probone{L}$, 
  establishing
  that $\probone{L}$ is \#P-hard, so it is not in FP because we assumed
  $\text{FP} \neq \#\text{P}$. Pick an arbitrary word $w \notin
  \LL(\Gamma)$. (Note that we do not care about the decidability of picking this
  constant word~$w$, as we are only showing that a reduction exists, not
  discussing the complexity of devising the reduction.)
  Given an input probabilistic word $p$ for $\probone{L_0}$, build the
  probabilistic word $w
  p$. This is in linear time, again because $w$ has constant size with respect
   to the
  input $p$. Now, we claim that the answer to $\probone{L_0}$ on~$p$ is the same
  as that of $\probone{L}$ on~$wp$, that is, an outcome $u$ of~$p$ is in~$L_0$
  iff the corresponding outcome $wu$ of~$wp$ (which has the same probability) is
  in~$L$. Now, from the definition of~$L$ and because the alphabets $\Sigma$ and
  $\Theta$ are disjoint, we know that the left residual
  $w^{-1} L = \{u \in (\Sigma \cup \Theta)^* \mid wu \in L\}$ is precisely
  $L_0$. Hence, $wu \in L$ iff $u \in L_0$, which is what we wanted to
  establish. Thus, the reduction is correct, which concludes
  the proof.
\end{proof}

In fact, the same proof shows that it is (conditionally) undecidable to decide,
given a CFL~$G$, whether $G$ is unambiguous,
poly-slicewise-unambiguous, or whether it admits tractable circuits (with or
without complement). Indeed, in the proof of the result above, in one case the
language to which we reduce is regular so it falls in all these classes, and in
the other case the language is intractable for $\prob$ so conditionally not in these
classes. Thus, it is an interesting question whether showing \#P-hardness of
$\prob$ can be a useful method to show (conditionally) that a language is
inherently ambiguous (compared to other methods~\cite{koechlin2022new}), or that
it is not poly-slicewise-unambiguous (compared to~\cite{mengel2025lower}), or
that it does not have tractable circuits.

Despite \cref{prp:meta}, one natural question
for future work is to classify more CFLs. Does the tractability
of $\pal^2$ extend to greater values of the \emph{palindromic length}
\cite{frid2013palindromic,borchert2015words}, e.g., $\pal^3$, or more generally
\emph{palstars}~\cite{knuth1977fast}? We can also see the languages $L_1$ of
\cref{prp:counthard} or $L_0'$ of \cref{prp:hard1b} as variants of palindrome
concatenations, with ``markers'' (the \#'s) and ``gaps'' (the $\Sigma^*$'s). Does
intractability still hold without these features? What about the language on
$\{0,1,\#\}$ defined as
$\{\Sigma^* \# u \# \Sigma^* \# u^R \# \Sigma^* \mid u \in \{0,1\}^*\}$? We
also do not know how tractability is affected if we instead study the restricted
case of $\prob$ corresponding to the counting of completions of partial words.
Another intriguing example is
\emph{Shamir's language}
\cite{koechlin2022new,shamir1971some}: $\prob$ for this language
amounts to computing, given two
probabilistic words $p$ and $p'$, the total probability of the outcomes $u,u'$
such that $u$ is a factor of~$u'$.

In the setting of counter automata, it is natural to wonder whether the
unambiguity requirement on counter automata is necessary to ensure that
probabilistic membership is tractable in data complexity. We have illustrated
some cases of ambiguous counter automata where the problem is intractable, but
it could be the case that tractability also holds in other setting. One
especially interesting question concerns vector addition systems with states
(VASS)~\cite{hopcroft1979reachability} with one counter, also called \emph{one-counter
nets}~\cite{hofman2016simulation}: for such counter automata, we conjecture that
probabilistic membership is tractable even in the presence of nondeterminism, by
a dynamic algorithm that keeps track of the probability of each function mapping
states to the maximal reachable counter value in a run that ends at this state.
We expect that this is in contrast with the intractability of the problem for
nondeterministic pushdown automata with unary stack alphabet or with
1-dimensional Parikh automata (see the end of \cref{sec:combined}); and expect
that tractability is also lost in the presence of more than one counter.
However, we do not know how this tentative tractable case fits in the
broader picture of our results.

In the setting of combined complexity, it remains a challenging question to
understand which restrictions on NFAs (or, indeed, on grammars) suffice to
ensure the tractability of the probabilistic membership problem in combined
complexity. In the setting of partial pattern matching, we leave open the
question of whether it is \#P-hard  to count the number of words
of length $n$ (given in unary) over the alphabet $\Sigma \coloneq \{0, 1\}$
having some factor matching a partial word $v$ over~$\Sigma$ given as input, in
particular in the case where $v$ only uses one letter (say $1$) and the
wildcard $?$ (remembering that \cite{ma2007complexity} shows this is NP-hard).

A broader direction is to understand which CFLs admit
tractable circuits in various circuit classes. For instance, when can we have
tractable \emph{decision diagrams} (e.g., OBDDs), or tractable \emph{structured
circuits}~\cite{pipatsrisawat2008new}? This relates to a line of work in 
probabilistic databases that asks which queries admit
tractable circuits in various representations~\cite{jha2013knowledge}, and more
broadly asks about the relative power of
circuit formalisms. One obvious question is whether
$\times,\uplus,\compl$-circuits capture the tractability of~$\prob$ for more
languages than $\times,\uplus$-circuits (i.e., whether they are more concise):
this relates to the open problem of whether d-DNNFs are closed under
complementation (see~\cite{darwiche2002knowledge,monet2020solving,vinall2024structured}).
Another question is whether there are CFLs that are tractable for $\prob$ but do
not admit tractable circuits: this relates to the question of whether such
circuits can express \emph{inclusion-exclusion}~\cite{amarilli2024non}, and
whether they can express the arithmetic computations used in the proof of 
\cref{prp:knfa}.

Finally, a much broader question 
the relationship of $\prob$
to non-probabilistic membership. Indeed, a central question in 
CFLs is to characterize the fine-grained (data) complexity of the (non-probabilistic)
membership problem for each specific
CFL $L$, i.e., the best achievable complexity in the input word. Lower bounds
are known for some grammars~\cite{abboud2018if}, and the ability to code
hard problems in $\probone{L}$ for a language~$L$ might relate to the ability to
code hard problems for parsing; so classifying the tractable CFLs
for $\prob$ may help with understanding the complexity of (non-probabilistic)
CFG parsing.

\section*{Acknowledgments}
  \noindent We are grateful to Arthur Ledaguenel for initial discussions
leading to the formulation of the problem; to Alexis de Colnet for early
discussions; to Louis Jachiet for remarking that the problem is tractable for
uCFLs and for pointing out an issue and a fix in the proof of
\cref{prp:pattern}; to Laurent Noé for pointers to related work; and to Charles Paperman and Michaël Cadilhac for advice on automata
models. We are also grateful to the anonymous referees of the conference version
of this paper for their useful feedback.

\bibliographystyle{plainurl}
\bibliography{main}

\end{document}